\documentclass[a4paper,onecolumn,11pt]{quantumarticle}
\pdfoutput=1
\usepackage{graphicx}% Include figure files
\usepackage{dcolumn}% Align table columns on
\usepackage{braket}
\usepackage{bm}% bold math
\usepackage{amsmath}
\usepackage{amsfonts}
\usepackage{mathtools}
\usepackage{amsthm}
\usepackage{lipsum}
\usepackage{tikz}
\usepackage{biblatex}
\newtheorem{theorem}{Theorem}
\newtheorem{lemma}{Lemma}
\newtheorem{corollary}{Corollary}
\newtheorem{proposition}{Proposition}

\newtheorem{definition}{Definition}
\newtheorem{remark}{Remark}

\bibliography{biblio}

\title{Hidden Variables for Pauli Measurements}

\author{Leon Bankston}
\affiliation{Tulane University}
\orcid{0000-0003-4439-3590}
\email{doogrammargood@gmail.com}
\begin{document}
\maketitle

\begin{abstract}
The Pauli measurements (the measurements that can be performed with Clifford operators followed by measurement in the computational basis) are a fundamental object in quantum information. It is well-known that there is no assignment of outcomes to all Pauli measurements that is both complete and consistent. 

We define two classes of hidden variable assignments based on relaxing either condition. Partial hidden variable assignments retain the consistency condition, but forfeit completeness. Contextual hidden variable assignments retain completeness but forfeit consistency. We use techniques from spectral graph theory to show that the incompleteness and inconsistency of the respective hidden variable assignments are both maximal.

As an application, we interpret our incompleteness result as a statement of contextuality and our inconsistency result as a statement of nonlocality. Our results show that we can obtain large amounts of contextuality and nonlocality using Clifford gates and measurements in the computational basis.
\end{abstract}

\section{Introduction}
Quantum information is concerned with the information-processing powers of quantum systems \cite{Nielsen16}. Since information is realized in physics \cite{Deutsch85}, these unique powers stem from the discrepancy between ``classical'' physics and quantum physics. 

There are several possible definitions for the classical systems. Two popular choices for properties defining ``classical'' behavior are locality \cite{scarani19} and noncontextuality \cite{Spekkens_2005} \cite{Cabello_2014} \cite{Acn2015}. Both concepts can be phrased in terms of their respective classes of hidden variables.

General hidden variables can trivially explain any phenomenon by appealing to predestination. To make a non-trivial claim, we constrain the hidden variables to obtain local hidden variables (that can be factored into hidden variables for various locations) and noncontextual hidden variables (that assign outcomes to measurements independent of which other compatible measurements occur with it).

The partial hidden variable assignments that we define in this paper are generalizations of (measurement) noncontextual hidden variables \cite{Spekkens_2005}, and the contextual hidden variable assignments are the factors of local hidden variables. Our definitions have the advantage that they come with values, $Pval$ and $Cval$, that measure the quality of the hidden variable assignments. Another advantage is conceptual. Our two types of hidden variable assignments are the extremes in the trade-off between the completeness and consistency of the assignment. 

We are interested a natural class of measurements that we call the Pauli measurements that are fundamental to quantum information. These measurements are also called stabilizer measurements \cite{Van_den_Nest_2004}. They can be implemented with existent hardware \cite{Niemann14}, and are the foundation for many protocols related to quantum cryptography, quantum error correction and quantum tomography. Our results add to our understanding of the structure of these essential measurements by showing that any assignment of outcomes to them must be very incomplete or very inconsistent.

The Pauli measurements are also interesting because they generalize many of the simple proofs that quantum physics is neither local nor noncontextual \cite{mermin93} \cite{mermin_2007}. On the other hand, the measurements can be simulated efficiently when applied to the computational zero state \cite{Aaronson_2004}. We present new, quantitative proofs of the contextuality and nonlocality of quantum mechanics using measurements that can be simulated classically. This points to an apparent discrepancy between a notion of ``classical'' based on simulation and a notion of ``classical''  based on locality or noncontextuality.

Despite the existence of classical simulation, contextuality \cite{amaral19} and nonlocality \cite{de_Vicente_2014} are resources for quantum information processing tasks. We provide simple, practically implementable, procedures for generating these resources. Therefore, we expect that our results will have further applications for quantum information than we discuss in Section \ref{applications_section}.

Our most important contribution is the recognition that the mathematics of spectral graph theory applies well to the Pauli measurements. Roughly, spectral graph theory \cite{shl06} characterizes how spread out graphs are, provided they have large symmetry groups. Such graphs can be defined naturally from the Pauli measurements so that structure in the graph reflects structure of the measurements. Beyond the mere bounds on $Cval$ and $Pval$, our approach reveals a powerful and previously-unstudied structure within the Pauli measurements. We expect that our techniques apply to other questions about the Pauli group.

Though the Pauli measurements are fundamental to quantum information, they have not been studied using the techniques in this paper. In fact, we introduce the term ``Pauli measurements'' because there does not seem to be a standard definition in the literature. One reason is that there are a variety of similar objects that can be dubbed Pauli (or stabilizer) measurements \footnote{For example, the Pauli measurements described in \cite{Abramsky_2017} roughly correspond to $\mathcal{L}^n_1$ in our work. Other natural choices include the collection of all projective measurements whose outcomes are stabilizer states, the collection of measurements protocols \cite{Acn2015} that allow adaptive measurement, and the collection of 
 Pauli measurements that involve ancillary qubits.}. The outcomes of our Pauli measurements are stabilizer codes \cite{gottesman97}, which have been studied extensively. The Pauli measurements are also the measurements described in \cite{Cabello_2010}.

The mathematical object $\mathcal{L}^n$ that represents the Pauli measurements was studied in \cite{Lazar2017AssociationSA}, and is related to dual polar graphs \cite{Bro89} and $q$-krawtchouk polynomials \cite{stanton_1984}. It is also related to polar spaces \cite{batten_1997} and Tits buildings.

Our incompleteness can be interpreted as an error-robust experimental verification of contextuality, similar to the proposal in \cite{Cabello_2010}. Our inconsistency result can be phrased in terms of a particular class of binary constraint games, studied in \cite{zhengfeng13}. We do not know of any other work that has studied these games, but they are natural to define.

The rest of the paper is organized as follows. In Section \ref{basics_section}, we define the Pauli measurements as the mathematical object $\mathcal{L}^n$ and describe how $\mathcal{L}^n$ corresponds to a collection of measurements in the standard formulation of quantum mechanics. We introduce our notions of partial and contextual hidden variable assignments, as well as metrics $Pval$ and $Cval$ for their completeness and consistency. In Section \ref{background_section} we give background on the main mathematical tools that we will need to prove our results. Both the incompleteness and inconsistency results rely on the theory of expander graphs. The inconsistency result also uses some theory from nonlocal games. Section \ref{main_results_section} gives our main results, Theorem \ref{incompletness_thm} and Theorem \ref{inconsistency_thm}. We give applications in Section \ref{applications_section} by interpreting Theorem \ref{incompletness_thm} as a statement of contextuality and Theorem \ref{inconsistency_thm} as a statement of nonlocality. Section \ref{conclusions_section} concludes with some discussion of future inquiries.

\section{Basic notions}\label{basics_section}
    In this section, we define the core notions of our theory. We define the Pauli measurements and the two classes of hidden variable assignments (partial and contextual) for explaining outcomes of systems of measurements.
\subsection{The Pauli Measurements}\label{pauli_measurements_main_section}

In this paper, we are interested in a particular set of of projective measurements, equipped with the partial order of fine-graining. All measurements in this work are projective. We use the term ``outcome'' of a measurement to refer to a maximal eigenspace of the associated Hermitian operator.

If $a,b$ are measurements, the partial order of fine-graining is $a\leq b$ if every outcome of $b$ is also an eigenspace of $a$. A semilattice \cite{birkhoff1967} \cite{DELSARTE1976} of measurements is a collection of projective measurements that is also a graded $\wedge$-semilattice under this ordering. Though our hidden variable assignments are defined generally, we will evaluate them for a particular family of semilattices of measurements that we call the Pauli measurements $\mathcal{L}^n$, defined in Definition \ref{Pauli_measurement_def}.

We denote the Pauli group (defined in Appendix \ref{Pauli_measurement_section}) on $n$ qubits by $\mathcal{P}(n)$. Each element of $\mathcal{P}(n)$ (other than multiples of $I$) is an operator on $\mathbb{C}^{2^n}$ with two orthogonal eigenspaces and can therefore be viewed as a $2$-outcome projective measurement. Any collection of pairwise commuting observables may be performed simultaneously, resulting in a new measurement.

The Pauli measurements $\mathcal{L}^n$ are the elements of the semilattice of commuting collections of $\mathcal{P}(n)$. This is equivalent to the collection of measurements obtained by performing a Clifford operation followed by a partial or full measurement in the computational basis. The equivalence follows from the fact that the Clifford group acts transitively on $\mathcal{P}(n)$ and measurement in the computational basis is a Pauli measurement.

We associate $\mathcal{P}(n)$ with  $\mathbb{Z}_2^{2n}$ using the map $coord: \mathcal{P}(n)\to \mathbb{Z}_2^{2n}$ defined by $i^{\omega}XZ(x)\mapsto x$. This identification motivates defining the Pauli measurements as follows.

\begin{definition}[Pauli measurements]\label{Pauli_measurement_def}
    Let $n\in \mathbb{N}$. Define the Pauli measurements on $n$ qubits to be the poset of isotropic subspaces of $\mathbb{Z}_2^{2n}$ ordered by inclusion.
    
    We denote this poset by $\mathcal{L}^n$.
\end{definition}

See Appendix \ref{Pauli_measurement_section} for basic definitions and a discussion of the equivalence between $\mathcal{L}^n$ and collections of commuting measurements in $\mathcal{P}(n)$. 

It is obvious that $\mathcal{L}^n$ is not just a poset, but also a $\wedge-$semilattice graded by dimension, where the $\wedge$ operation is intersection. We denote the $k$-dimensional isotropic subspaces of $\mathbb{Z}^{2n}_2$ by $\mathcal{L}^n_k$ and the semilattice of subspaces of dimension $k$ or smaller by $\downarrow \mathcal{L}^n_k$.

\subsection{Partial Hidden Variables}

We formally define partial hidden variable assignments for semilattices of measurements and define $Pval$, our metric for evaluating their completeness. First, we define the key property of these assignments, consistency.

\begin{definition}[Consistency]\label{consistency_def}
    Let $\mathcal{M}$ be a semilattice of measurements. Let $m_1, m_2 \in \mathcal{M}$ and let $o_1$ and $o_2$ be eigenvalues for $m_1$ and $m_2$ respectively. For any $m_3\leq m_1 \wedge m_2$, $m_3$ inherits an eigenvalue $o_1(m_3)$ from $o_1$ and an eigenvalue $o_2(m_3)$ from $o_2$, due to the definition of the partial order. Define the outcomes $o_1$ and $o_2$ to be consistent at $m_3$ if $o_1(m_3)=o_2(m_3)$. 
    
    The two outcomes $o_1$ and $o_2$ are called consistent if they are consistent at every $m_3\leq m_1\wedge m_2$. A collection of outcomes is called consistent if every pair of outcomes is consistent.
\end{definition}

\begin{definition}[Partial hidden variable assignment]
     Let $\mathcal{M}$ be a semilattice of measurements, and let $M\subset \mathcal{M}$ be the set of maximal measurements. A partial hidden variable assignment for $\mathcal{M}$ is a partial function $f$ on $M$ that takes some maximal measurements to their respective outcomes such that the image of $f$ is consistent.
\end{definition}

Partial hidden variable assignments can be ranked according to the size of their domains. We are interested in those assignments that have large domains, so we introduce notation to describe this size.

\begin{definition}[Pval]
     Let $\mathcal{M}$ be a semilattice of measurements whose maximal elements are $M\subset \mathcal{M}$ and let $f$ be a partial hidden variable assignment. Then $Pval(f)\coloneqq \frac{|Dom(f)|}{|M|}$, where $Dom(f)$ is the domain of $f$.
     
     Define $Pval(\mathcal{M})$ to be $\sup_f(Pval(f))$, where $f$ runs over all partial hidden variable assignments for $\mathcal{M}$.
\end{definition}

We will show (as Theorem \ref{incompletness_thm}) that $Pval(\mathcal{L}^n)$ decreases to $0$ exponentially in $n$.

\subsection{Contextual Hidden Variables}
Here, we define contextual hidden variables and define $Cval$, our criterion for evaluating their consistency. 

\begin{definition}[Contextual hidden variable assignment]
     Let $\mathcal{M}$ be a semilattice of  measurements. A contextual hidden variable assignment is an assignment of outcomes to maximal measurements of $\mathcal{M}$.
\end{definition}

We can evaluate contextual hidden variable assignments based on their consistency.
\begin{definition}[Cval]\label{cval_def}
     Let $\mathcal{M}$ be a semilattice of measurements. Let $f$ be a contextual hidden variable for $\mathcal{M}$. Let $w\in \mathcal{M}$ be a minimal nontrivial \footnote{By nontrivial, we mean that there are at least $2$ distinct outcomes.} measurement. Let $M_w\subset \mathcal{M}$ be the set of maximal measurements above $w$.
     
     For each $x\in M_w$, assign an outcome to $w$ by inheritance from $f(x)$. Let $m_w$ be the fraction of $x\in M_w$ that assign the minority outcome. Define $Cval(f,w)\coloneqq m_w(1-m_w)$.
     
     Define $Cval(\mathcal{M})\coloneqq \min_{f} \mathbb{E}_w [Cval(f,w)]$. Here, $f$ runs over all contextual hidden variables for $\mathcal{M}$.
\end{definition}

We will show (as Theorem \ref{inconsistency_thm}) $\lim_{n\to \infty} Cval(\mathcal{L}^n_n) = \frac{1}{4}$, which is the largest possible value.

\section{Background}\label{background_section}
    We describe some of the background for the technical tools that we will use to prove our incompleteness and inconsistency results, Theorem \ref{incompletness_thm} and Theorem \ref{inconsistency_thm}.

\subsection{Spectral Graph Theory}\label{spec_graph_theory_subsection}
    In this subsection, we describe the graph-theoretic tools that are central to both our incompleteness and inconsistency results.
    
    The graphs that appear in this paper are finite simple graphs. For basic terminology, see \cite{Biggs93} or \cite{Godsil01}. We will define graphs using the Pauli measurements, and structure in the graphs will reflect structure of $\mathcal{L}^n$.
    
    The top-level Pauli measurements $\mathcal{L}^n_n$ have a natural distance that measures the amount of common information collected between pairs of Pauli measurements.
    
    \begin{definition}[Distance]\label{pauli_distance_def}
    Define $d: \mathcal{L}^n_n\times \mathcal{L}^n_n \to \{0,\dots,n\}$ by $(x,y)\mapsto n-\dim(x\cap y)$.
    \end{definition}

    It is routine to see that $d$ satisfies the usual properties of distance.\footnote{In fact, $d$ is the graph distance of the graph with vertex set $\mathcal{L}^n_n$ with $x\sim y$ iff $d(x,y)=1$.}
    
    Our main technical tool is the fact that $\mathcal{L}^n_n$ with this distance defines an association scheme \cite{Bro89}, so it is possible to obtain combinatorial data \cite{shl06} about the Pauli measurements based on the spectrum of the graph.
    
    \begin{definition}[Spectrum of a graph]\label{graph_spectrum_def}
        Let $G$ be a graph. 
        Its spectrum $Spec(G)$ is the set of eigenvalues of the matrix $(A_{v,w})_{v,w\in V(G)}$ where $A_{v,w} = 1$ if $v\sim w$ and $0$ otherwise.
        
        We will call the second largest element of $\{|l| \mid l \in Spec(G)\}$ the spectral parameter and denote it by $\lambda$.

        We will call $\frac{\lambda}{\max(Spec(G))}$ the spectral ratio.
    \end{definition}

    We will use the following $3$ families of graphs.
    
    \begin{definition}[$G^\prime_w$]\label{G1graph_def}
        Let $n\in \mathbb{N}$. Let $w\in \mathcal{L}^n_1$. Define the graph $G^\prime_w$ to have vertex set $V(G^\prime_w)=\mathcal{L}^n_n \cap \uparrow w$, with $x\sim y$ iff $d(x,y)=1$.
    \end{definition}

    \begin{definition}[$G_w$]\label{G2graph_def}
         Let $n\in \mathbb{N}$ be even. Define the graph $G_w$ to have vertex set $V(G_w)=V(G^\prime_w)$
         with $x\sim y$ iff $d(x,y)=\frac{n}{2}$.
    \end{definition}

    \begin{definition}[$B_{n,2}$]\label{Bgraph_def}
        Let $n\geq 3$. Define the bipartite graph $B_{n,2}$ to have vertex set $V(B_{n,2}) = \mathcal{L}^n_n \cup \mathcal{L}^n_2$. Edges are given by strict containment.
    \end{definition}
    
    The source of our results is the observation that families of regular graphs with small spectral ratios are well-spread out, as formalized by the following lemmas.
    
    \begin{lemma}[Expander mixing lemma]\label{expander_mixing_lemma}(\cite{vad12}, Theorem $4.15$)
    Let $G$ be a regular connected graph. Let $\lambda$ be the spectral parameter of $G$.
    
    Suppose $S,T \subset V(G)$. Then
    
    \[| E(S,T) - \Delta(G)\frac{|S||T|}{|V(G)|} | \leq \lambda \sqrt{|S||T|}\]
    
    where $E(S,T)$ is the number of pairs $(v,w)\in S\times T$ such that $v\sim w$.
\end{lemma}

A similar result, proven in \cite{HAEMERS}, holds for biregular graphs and will be useful for our incompleteness result.

\begin{lemma}[Bipartite expander mixing lemma]\label{biregular_expander_mixing_lemma}
    Let $B$ be a biregular graph, i.e. a bipartite graph $V(B)=R\cup L$, where each $r\in R$ has a common degree $\Delta(R)$ and each $l\in L$ has a common degree $\Delta(L)$. Let $\lambda$ be the spectral parameter of $B$.
    
    Suppose $S\subset L$ and $T\subset R$, and that $|S|=\alpha|L|$ and $|T|=\beta|R|$, for some $\alpha, \beta \in [0,1]$. Then
    
    \[
    |\frac{E(S,T)}{E(G)} - \alpha \beta| \leq \frac{\lambda}{\sqrt{\Delta(R)\Delta(L)}} \sqrt{\alpha \beta (1-\alpha)(1-\beta)}
    \]
    
    where $E(S,T)$ is the number of edges between $S$ and $T$.
\end{lemma}

Another natural metric for the spread of a graph is given by how evenly random walks explore the vertices. One manifestation of this, needed in our applications section, is called the Hitting Lemma. We state a version that can be proven easily by slightly modifying Theorem $4.17$ of \cite{vad12}.

\begin{lemma}[Hitting lemma]\label{hitting_lemma}
    Let $G$ be a regular graph with spectral parameter $\lambda$. Let $\{A_i\}_{i=1}^k$ be subsets of $V(G)$ with $\frac{|A_i|}{|V(G)|}\leq \mu$ for each $i\in \{1,\dots,k\}$, where $\mu \in (0,1)$. Let $(v_1,\dots,v_k)$ be a walk on $G$ chosen uniformly at random. Then
    \[
    Pr(v_i \in A_i, \forall i \in \{1,\dots, k\}) \leq  (\mu + \frac{\lambda}{\Delta(G)}(1-\mu))^k.
    \]
\end{lemma}

In order to apply these lemmas, we need to know the spectral ratios for the three graphs $G^\prime_w$, $G_w$, and $B_{n,2}$.

The first proposition shows that we cannot apply the expander mixing lemmas to $G^\prime_w$.
\begin{proposition}\label{full_graph_spectrum_prop}
    Let $R(G^\prime_w)$ be the spectral parameter of $G^\prime_w$ (which implicitly depends on $n$). Then $\log(R(G^\prime_w)) = -1 + o(1)$.
\end{proposition}

However, we can apply the expander mixing lemmas to the other two families of graphs.
\begin{proposition}\label{half_graph_spectrum_prop}
    Let $R(G_w)$ be the spectral ratio of $G_w$. Then $\log (R(G_w)) = -\frac{1}{2}n$.
\end{proposition}
\begin{proposition}\label{bipartite_spectrum_prop}
    Let $R(B_{n,2})$ be the spectral ratio of $B_{n,2}$. Then
$\log(R(B_{n,2})) \leq -\frac{1}{2}n+O(1)$.
\end{proposition}

Proofs can be found in Appendix \ref{eigenvalue_conjectures_section}.
\subsection{Nonlocal Games}
In this subsection, we give some background information about nonlocal games \cite{Brunner_2014}. They will be critical tools for our inconsistency result.

\begin{definition}[Nonlocal game]
        A (2-player) nonlocal game consists of a collection of questions for each player, $X,Y$ and a set of allowed answers, $A$ and $B$.
        
        A referee chooses $(x,y)\in X\times Y$ according to a probability distribution $\pi:X\times Y \to [0,1]$. Given $x$ and a resource shared with Bob, Alice provides a response $a\in A$ without knowing $y$. Similarly, Bob provides a response $b\in B$ given access only to $y$ and a shared resource with Alice.
        
        A (possibly probabilistic) predicate $V(a,b,x,y) \in [0,1]$ describes whether Alice and Bob win ($V(a,b,x,y)=1$) or lose ($V(a,b,x,y)=0$) given their questions and answers.
\end{definition}

\begin{remark}
    All of the nonlocal games in this paper will have $X=Y$, with both being a set of Pauli measurements at a given level, and $A=B$ being the set of  allowed outcomes on those measurements. Though technically each measurement has its own set of outcomes, measurements at a given level all have the same number of outcomes. This means that we can number them arbitrarily and use this numbering as a canonical set of responses, $A$, independent of the particular $x\in X$ that Alice is given as a question.
\end{remark}

\begin{definition}[Strategy types]
     Let $Z$ be a nonlocal game with question sets $X,Y$ and answer sets $A,B$.
     \begin{enumerate}
     
     \item A local ($loc$) strategy is one in which Alice and Bob's shared resource consists of classical information.
     \item A deterministic strategy ($det$) is one in which Alice and Bob's answers are functions of  their questions.
     \item If $X=Y$, a synchronous strategy ($syn$)  \cite{Dykema2019} is a deterministic strategy where Alice and Bob adopt the same strategy.
     \item A quantum strategy $(qtm)$ allows Alice and Bob to share parts of an entangled quantum state as their shared resource. They may perform local measurements and use the results of their measurement to decide on their responses.
     \end{enumerate}
     
     For each of the $4$ types above $r\in \{loc,det,syn,qtm\}$, we define $Val_{r}(Z)$ to be the optimal probability of winning when the strategies come from their respective sets.
\end{definition}

It is well-known that $Val_{loc}(Z)=Val_{det}(Z)$ for any game $Z$. Obviously, $Val_{syn}(Z) \leq Val_{det}(Z) \leq Val_{qtm}(Z)$ because every synchronous strategy is also a deterministic strategy and every deterministic strategy is also a quantum strategy.

Observe that a deterministic strategy can be described by a pair of contextual hidden variable assignments. For a synchronous strategy, Alice and Bob adopt the same strategy, so only one contextual hidden variable theory is needed.

The following construction allows us to create new games from old ones by playing them in parallel.

\begin{definition}[Parallel repetition]\label{parallel_repetition_def}
     Let $Z$ be a nonlocal game with question sets $X$ and $Y$, question distribution $\pi$ and answer sets $A$ and $B$. For $k\in \mathbb{N}$, define $Z^{\otimes k}$ to be the game whose questions are $k$-tuples of questions from $Z$, $X^k$ and $Y^k$. The questions are distributed according to the product distribution $\pi^k$, meaning that the question for each coordinate are drawn uniformly at random according to $\pi$.
     
     Alice's responses are $k$-tuples $A^k$. That is, she provides one response for each coordinate of her questions, and her $i^{th}$ response must be a valid response in $Z$ to her $i^{th}$ question. Bob's responses in $Z^{\otimes k}$ are defined similarly as $k$-tuples, $B^k$.
     
     Alice and Bob win the parallel repetition game $Z^{\otimes k}$ if they win on each coordinate according to the predicate for $Z$.
\end{definition}

The parallel game $Z^{\otimes k}$ is subtly different from playing $Z$ $k$-times independently in that Alice and Bob are able to coordinate their responses within a list of $k$ questions posed in parallel. This subtlety makes the following theorem non-obvious \cite{lovasz92}.

\begin{lemma}[Parallel repetition theorem]\label{parallel_repetition_lemma}(\cite{Holenstein_2009}, Theorem $4$.)
    Let $Z$ be a nonlocal game with answer sets $A$ and $B$. Then
    
    \[Val_{loc}(Z^{\otimes k}) \leq \left(1-\frac{(1-Val_{loc}(Z))^3}{6000}\right)^{\frac{k}{\log(|A||B|)}}.\]
\end{lemma}

\section{Main results}\label{main_results_section}

In this section, we derive our main results. Theorem \ref{incompletness_thm} shows that partial hidden variable assignments must be very incomplete, and Theorem \ref{inconsistency_thm} shows that contextual hidden variable assignments must be very inconsistent. These results can be viewed as a generalization of the basic fact that the Pauli measurements do not have an assignment of variables that is both complete and consistent \cite{mermin93}. We recover this result as Corollary \ref{incomplete_inconsistent_cor}.

\subsection{Incompleteness Bound}
    In this section, we describe our incompleteness result. Our strategy is to work towards a bound on $Pval(\mathcal{L}^n)$ by calculating $Pval(\downarrow \mathcal{L}^n_2)$, then apply structure from spectral graph theory to arrive at an asymptotic bound as $n\to \infty$.

    With the help of Proposition \ref{lin_antilin} in Appendix \ref{Pauli_measurement_section}, we can formulate $Pval(\mathcal{L}^2)$ as a MAXSAT instance and calculate it directly. Our result agrees with \cite{Cabello_2010}.
    \begin{proposition}\label{L22_prop}
        $Pval(\mathcal{L}^2)=\frac{12}{15}=\frac{4}{5}.$
    \end{proposition}

    The following proposition is proven in Appendix \ref{supp_proofs_section} using an inductive argument.
    \begin{proposition}\label{L22_lemma_prop}
        $Pval(\downarrow \mathcal{L}^n_2)\leq \frac{4}{5}$.
    \end{proposition}

    \begin{corollary}\label{incomplete_inconsistent_cor}
    $\mathcal{L}^n$ does not have a hidden variable assignment that is both complete and consistent.
    \end{corollary}
    \begin{proposition}\label{partial_var_prop}
    $Pval(\mathcal{L}^n)\leq 4 (\frac{
    \lambda}{\Delta(B_{n,2})})^2$, where $\lambda$ is the spectral parameter of $B_{n,2}$. See subsection \ref{spec_graph_theory_subsection} for definitions.
    
	\end{proposition}

    \begin{proof}
        Let $f$ be a partial hidden variable assignment for $\mathcal{L}^n$. Then $f$ restricts to a partial hidden variable assignment $\downarrow \mathcal{L}^n_2$. We use this restriction to argue that if $Pval(\mathcal{L}^n)$ exceeds our bound, then $Pval(\downarrow \mathcal{L}^n_2)$ is large enough to contradict Proposition \ref{L22_lemma_prop}.

        The vertices of $B_{n,2}$ in $\mathcal{L}^n_2$ that inherit an outcome from restriction of $f$ are the neighbors of $Dom(f)\subset \mathcal{L}^n_n$. 
        
        We use Lemma \ref{biregular_expander_mixing_lemma} to estimate the size of the neighborhood of $Dom(f)$ in terms of $|Dom(f)|$ and the square of the spectrum of $B_{n,2}$.
         
         See Appendix \ref{supp_proofs_section} for details.
    \end{proof}
    
    Our incompleteness bound follows immediately from Proposition \ref{bipartite_spectrum_prop}, applied to Proposition \ref{partial_var_prop}.
    
    \begin{theorem}[Incompleteness bound]\label{incompletness_thm}

		$Pval(\mathcal{L}^n) \leq 2^{-n +O(1)}$.
	\end{theorem}
	
	Thus, the fraction of measurements of $\mathcal{L}^n_n$ that can be given outcomes consistently approaches $0$ exponentially in $n$.
    
\subsection{Inconsistency Bound}
    Our inconsistency bound also relies on eigenvalue techniques to extend simple results for $\mathcal{L}^2$ to the more complicated case of $\mathcal{L}^n$. We define a series of games, $Z_1$, $Z_1^{\otimes n}$, $\bar{Z}_\frac{n}{2}$ and $Z_{\frac{n}{2}}$. Applying Lemma \ref{expander_mixing_lemma} to the last game  $Z_{\frac{n}{2}}$ derives the main result of this subsection, the inconsistency bound, Theorem \ref{inconsistency_thm}.

    Let $Z_1$ be the nonlocal game whose questions are pairs $(x,y)\in\mathcal{L}^2_2 \times \mathcal{L}^2_2$, such that $\dim(x\cap y) =1$. The questions are asked uniformly at random from this set. Alice and Bob are challenged to provide outcomes for their measurements that agree on their common measurement, $x\cap y \in \mathcal{L}^2_1$.
    
    Using quantum resources, Alice and Bob can always win $Z_1$. They can employ a similar strategy to win all games defined in this paper when allowed quantum resources.
    \begin{proposition}\label{qtm_always_wins_prop}
        $Val_{qtm}(Z_1) = 1$.
    \end{proposition}
    \begin{proof}
        Suppose Alice and Bob share the canonical maximally entangled pair  $\ket{\psi}=\frac{1}{\sqrt{2}}\sum_{i=1}^{2} \ket{e_i}\ket{e_i}$ and are challenged with the questions $x, y\in \mathcal{L}^2_1$ respectively. We will provide a quantum strategy that allows Alice and Bob to win $Z_1$ with certainty.
        
        Alice will perform the set of local measurements $\{\widetilde{XZ}(a) \mid a \in x\}$ on her half of $\ket{\psi}$, and Bob will perform the set of local measurements $\{\widetilde{XZ}(b) \mid b \in y\}$ on his half of $\ket{\psi}$. Alice will submit her outcome as her response. Bob will transform his outcome before submitting it. Specifically, if his outcome is given by the linear/antilinear function $f:y\to \mathbb{Z}_2$, he will submit the linear/antilinear function $\tilde{f}:y\to \mathbb{Z}_2$ given by $b=(b_1,b_2)\mapsto (-1)^{b_1 \cdot b_2}f(b)$. It remains to show that this strategy wins $Z_1$ with certainty. We show this in Appendix \ref{supp_proofs_section}.
    \end{proof}
    
    Without quantum resources, Alice and Bob's individual strategies can each be described by a contextual hidden variable assignment for $\mathcal{L}^2$. First, we show that they cannot always win $Z_1$ with a synchronous strategy.
    
    \begin{proposition}
    
    There exists $\epsilon>0$ such that $1-Val_{syn}(Z_1)\geq \epsilon$.
    \end{proposition}
\begin{proof}
   We describe Alice and Bob's common strategy with a  contextual hidden variable assignment for $\mathcal{L}^2$. If Alice and Bob always win $Z_1$, then this hidden variable assignment has no contradictions, so it is both a complete and partial hidden variable assignment for $\mathcal{L}^2$. This contradicts Corollary \ref{incomplete_inconsistent_cor}.
\end{proof}
    The next proposition shows that limitations on synchronous strategies for $Z_1$ imply limitations on deterministic strategies for $Z_1$. A similar fact holds for all games defined in this paper.
    \begin{proposition}
        Let $\epsilon >0$. If $1- Val_{syn}(Z_1) \geq \epsilon >0 $ then $1-Val_{loc}(Z_1) \geq \frac{\epsilon}{2}$.
    \end{proposition}
    \begin{proof}
        Let us referee $Z_1$ by following a protocol. All choices are made uniformly at random over their allowed sets.
        \begin{enumerate}
            \item Choose $w\in \mathcal{L}^2_1$.
            \item Choose $b\in \mathcal{L}^2_2$ such that $b\geq w$.
            \item Choose $a\in \mathcal{L}^2_2$ such that $a\geq w$ and $a\neq b$.
            \item Choose $a^\prime\in \mathcal{L}^2_2$ such that $a^\prime \geq w$, $a^\prime \neq a$ and $a^\prime \neq b$.
            \item Give Bob $b$ and choose which of $a$ or $a^\prime$ to give Alice.
        \end{enumerate}
        
        After performing these $5$ steps, Alice and Bob will receive questions according to the distribution for $Z_1$. Also, if we marginalize out the choices for $w$ and $b$, then $a$ and $a^\prime$ are both drawn uniformly at random according to the question set for $Z_1$. Focusing on $a$ and $a^\prime$, we can view Alice as playing $Z_1$ against herself. Thus, with probability at least $\epsilon$, Alice's responses to the Pauli measurement $w$ will be different under the choice of $a$ in step $5$ versus the choice of $a^\prime$. When Alice's responses differ, Bob has at most a probability of $\frac{1}{2}$ of agreeing with Alice. The protocol makes it clear that Bob must disagree with Alice probability at least $\frac{\epsilon}{2}$.
    \end{proof}

    The game $Z_1^{\otimes n}$ is defined as the $n$-fold parallel repetition of $Z_1$ (see Definition \ref{parallel_repetition_def}). $Z_1^{\otimes n}$ can still be played perfectly with a quantum strategy. Lemma \ref{parallel_repetition_lemma} bounds $Val_{loc}(Z_1^{\otimes n})$ by exponential decay in $n$. Specifically, we have the following proposition.
    
    \begin{proposition}\label{parallel_prop}
        For some $\epsilon >0 $, $Val_{loc}(Z^{\otimes n}_1) \leq (1-\frac{\epsilon^3}{2^3 6000})^\frac{n}{8}.$
    \end{proposition}
    
   Now we define our final game in the series.

    \begin{definition}
        For $n$ even, let $Z_\frac{n}{2}$ be the nonlocal game whose questions are pairs $(x,y)\in \mathcal{L}^n_n \times \mathcal{L}^n_n$ such that $d(x,y) = \frac{n}{2}$ and answers are outcomes of those measurements. The winning condition is that Alice and Bob's answers must be consistent on their common Pauli measurements.
    \end{definition}
    As in Proposition \ref{qtm_always_wins_prop}, it is trivial that $Val_{qtm}(Z_{\frac{n}{2}}) = 1$. The local value can be bounded by using an embedding, analogous to Proposition \ref{L22_lemma_prop}.
    
    \begin{proposition} \label{parallel_nonparallel_prop}
        $Val_{loc}(Z_{\frac{n}{2}}) \leq Val_{loc}(Z^{\otimes n}_1)$.
    \end{proposition}
    \begin{proof}
        Let us define an auxiliary game, $\bar{Z}_{\frac{n}{2}}$ which is the same as $Z_{\frac{n}{2}}$ except that Alice and Bob are given a hint $\phi\in Sp(2n,\mathbb{Z}_2)$. The hint $\phi$ transfers their pair of questions from $Z_\frac{n}{2}$ to a pair of questions of $Z^{\otimes n}_1$. The hint $\phi$ is to be chosen uniformly at random with this property.
        
        Clearly, $Val_{loc}(Z_{\frac{n}{2}}) \leq Val_{loc}(\bar{Z}_{\frac{n}{2}})$, because Alice and Bob can always ignore the hint.
        
        We obtain another inequality $Val_{loc}(\bar{Z}_{\frac{n}{2}})\leq Val_{loc}(Z^{\otimes n}_1)$ because any strategy for $Z_1^{\otimes n}$ can be used as a strategy for $\bar{Z}_{\frac{n}{2}}$ by using $\phi$ to convert the questions for $\bar{Z}_{\frac{n}{2}}$ to questions for $Z^{\otimes n}_1$. Note that this conversion respects question distributions.

        Chaining the two inequalities gives the result.
    \end{proof}

    By chaining the inequalities of Proposition \ref{parallel_nonparallel_prop} and Proposition \ref{parallel_prop}, we have
    \begin{corollary}\label{half_game_cor}
        There exists $\epsilon \in (0,1)$ such that $Val_{syn}(Z_\frac{n}{2})\leq Val_{loc}(Z_\frac{n}{2})< \epsilon ^n$ for all even $n \in \mathbb{N}$.
    \end{corollary}
    We can relate $Val_{loc}(Z_{\frac{n}{2}})$ to $Cval(\mathcal{L}^n)$ by counting the following objects in two ways.
    
    \begin{definition}[Contradiction triangle]
         Let $n\in \mathbb{N}$ be even. Fix a contextual hidden variable assignment $f$ for $\mathcal{L}^n$. A contradiction triangle is a triple $(w,x,y)$ with $w\in \mathcal{L}^n_1$, $x,y\in \mathcal{L}^n_n$, such that $d(x,y)=\frac{n}{2}$, and such that $f(x)$ and $f(y)$ are not consistent at $w$.

         For a given $n$ and contextual hidden variable assignment for $\mathcal{L}^n$, number of contradiction triangles will be denoted $T_n$.
    \end{definition}
    
    We first count the number of contradiction triangles $T_n$ by counting the number of pairs $(x,y) \in \mathcal{L}^n_n \times \mathcal{L}^n_n$ such that $d(x,y) = \frac{n}{2}$ and $f(x)$ is inconsistent with $f(y)$. We then count the number of $w$ that complete the triangle using Corollary \ref{disagreement_on_half_cor} to obtain the following proposition.
    \begin{proposition}\label{triangle_count_prop}
         Fix a contextual hidden variable assignment for $\mathcal{L}^n$. The number of contradiction triangles is at least $T_n\geq (1-Val_{syn}(Z_{\frac{n}{2}}) ) Q 2^{\frac{n}{2}-1}$, where $Q$ is the number of pairs of questions in the game $Z_{\frac{n}{2}}$.
    \end{proposition}
    
    We use eigenvalue techniques to bound the number of contradiction triangles containing a given $w\in \mathcal{L}^n_1$.
    
    \begin{proposition}\label{contradiction_triangles_containing_w}
        Fix a contextual hidden variable assignment $f$ for $\mathcal{L}^n$ where $n$ is even. Fix $w\in \mathcal{L}^n_1$.
        Let $G_w$ be the graph defined in Definition \ref{G2graph_def}. Let $m_w \in [0, \frac{1}{2}]$ be the fraction of $V(G_w)$ such that $f(x)$ restricted to $w$ assigns $w$ its minority outcome, as in Definition \ref{cval_def}.
        
        The number of contradiction triangles containing $w$ is at most
        
        \[ T_n\leq 2|V(G_w)|\left(\Delta(G_w) m_w (1-m_w) + \lambda \sqrt{m_w (1-m_w)}\right)
        \]
        
        where $\Delta(G_w)$ is the degree of $G_w$ and $\lambda$ is the spectral parameter of $G_w$.
    \end{proposition}
    
    \begin{proof}
        Color $V(G_w)$ according to the outcome of $f$ when restricted to $w$. A contradiction triangle containing $w$ corresponds to an adjacent pair $x\sim y$ such that $x$ and $y$ have different colors.
        
       We use Lemma \ref{expander_mixing_lemma} to estimate the number of such pairs using only structural constants of $G_w$, $\Delta(G_w)$ and $\lambda$, along with the sizes of the sets with each color, $m_w$ and $1-m_w$.
    \end{proof}
    
    \begin{corollary}\label{num_triangle_cor}
    Fix a contextual hidden variable assignment for $\mathcal{L}^n$ where $n \geq 2$. The total number of contradiction triangles is at most 
    \[
    T_n\leq 2|V(G_w)|\Delta(G_w) \sum_{w\in \mathcal{L}^n_1}m_w (1-m_w) + \lambda|\mathcal{L}^n_1| |V(G_w)|.
    \]
    \end{corollary}
    \begin{proof}
        Use Proposition \ref{contradiction_triangles_containing_w} and sum over $w\in \mathcal{L}^n_1$ and use that $\sqrt{m_w(1-m_w)}\leq \frac{1}{2}$.
    \end{proof}
    
    \begin{lemma}\label{vertex_subgraph_count_lemma} Let $n\geq 2$ and $w\in \mathcal{L}^n_1$. Then
        \[|V(G_w)| = \frac{|\mathcal{L}^n_n| (2^n-1)}{|\mathcal{L}^n_1|}.\]
    \end{lemma}
    \begin{proof}
        We count the number of pairs $(w,x)$ with $w\in \mathcal{L}^n_1$ and $x\in \mathcal{L}^n_n$ such that $w\leq x$ in two ways. Either choose $w$ first, resulting in a count of $|\mathcal{L}^n_1||V(G_w)|$ or choose $x$ first which results in a count of $|\mathcal{L}^n_n| (2^{n}-1)$.
    \end{proof}
    
    \begin{lemma}\label{cval_approximation_lemma}
    Let $n$ be even. Then
    \[
    \mathbb{E}_{w\in \mathcal{L}^n_1}[m_w(1-m_w)]\geq \frac{(1-Val_{syn}(Z_{\frac{n}{2}}))Q2^{\frac{n}{2}-1}}{2|\mathcal{L}^n_n|(2^n-1)\Delta(G_w)}
    -\frac{\lambda}{2\Delta(G_w)}.
    \]

    where $Q$ is the number of questions of $Z_{\frac{n}{2}}$, $\lambda$ is the spectral parameter of $G_w$, and the expected value draws $w\in \mathcal{L}^n_1$ uniformly at random.
    \end{lemma}
    \begin{proof}
        Bound the count of contradiction triangles in Proposition \ref{triangle_count_prop} via Corollary \ref{num_triangle_cor}, then replace $|V(G_w)|$ by applying Lemma \ref{vertex_subgraph_count_lemma}.
    \end{proof}
    
    The following lemma, proven in Appendix \ref{degree_count_section}, gives formulas the relevant quantities that appear in Lemma \ref{cval_approximation_lemma}.
    \begin{lemma}\label{cval_approximation_helper_lemma}
    \begin{enumerate} Let $n\in \mathbb{N}$ be even, $k\leq n$, and $Q$ be the number of questions of $Z_{\frac{n}{2}}$.
        \item \[|\mathcal{L}^n_k| = \prod_{i=0}^{k-1} \frac{2^{2n-i}-2^i}{2^{k}-2^i}.\]
        
        \item \[
        Q={n \choose \frac{n}{2} }_2 2^{\frac{n}{2}+1 \choose 2} |\mathcal{L}^n_n|. 
        \]
        \item \[
        \Delta(G_w) = 2{n-1\choose \frac{n}{2}-1}_2 2^{\frac{n}{2}+1 \choose 2}.
        \]
        \item \[
        \frac{Q}{|\mathcal{L}^n_n|\Delta(G_w)} = \frac{2^{\frac{n}{2}+1}-1}{2^{\frac{n}{2}}-1}\prod_{i=0}^{\frac{n}{2}-2} \frac{2^{n-i}-1}{2^{n-i-1}-1} \geq 2^{\frac{n}{2}}.
        \]
        \end{enumerate}
    \end{lemma}

   We apply Lemma \ref{cval_approximation_helper_lemma} to Lemma \ref{cval_approximation_lemma} and simplify.
    \begin{proposition}\label{cval_prop}
    \[Cval(\mathcal{L}^n) = \mathbb{E}_{w\in \mathcal{L}^n_1}[m_w(1-m_w)]\geq \frac{1-Val_{syn}(Z_{\frac{n}{2}})}{4} - \frac{\lambda}{2\Delta(G_w)}.\]

    The expected value draws $w\in \mathcal{L}^n_1$ uniformly at random.
    \end{proposition}

    Finally, we apply Proposition \ref{half_graph_spectrum_prop} and obtain the inconsistency bound.
    
    \begin{theorem}[Inconsistency bound]\label{inconsistency_thm}
        There exists  $\epsilon \in (0,1)$ such that for each even $n\in \mathbb{N}$, \[Cval(\mathcal{L}^n) \geq \frac{1}{4} - \epsilon^n.\]
    \end{theorem}

\section{Applications}\label{applications_section}
    In this section, we expand on the idea that bounds on partial hidden variable assignments for $\mathcal{L}^n$ can be interpreted as contextuality, as introduced in \cite{Cabello_2014} and formalized in \cite{Acn2015}. Roughly speaking, contextuality is the idea that outcomes of measurements are sensitive to which other measurements are performed simultaneously.
    
    In this formalism, contextuality is described by a graph $G$ such that $\vartheta(G) >\alpha(G)$, where $\vartheta(G)$ is the Lovasz-theta number defined in \cite{Lovasz79} and $\alpha(G)$ is the independence number.
    
    Let us recall the definition of $\vartheta(G)$, since it is relevant to contextuality.
    
    \begin{definition}[Orthonormal representation]
    Let $G$ be a graph. An orthonormal representation is a function $f:V(G)\to \mathbb{C}^n$ (for some $n\in \mathbb{N}$) that sends vertices of $G$ to unit vectors such that if $v \sim w$ then $\left<f(v), f(w)\right>=0$.
    \end{definition}
    \begin{definition}\label{theta_definition}
         Let $G$ be a graph.
         
         Then $\vartheta(G)$ is defined as
         
         \[\vartheta(G)\coloneqq \max_{f,\ket{\psi}} \sum_{v\in V(G)} |\braket{\psi|f(v)}|^2\]
         
         where $f$ ranges over all orthonormal representations of $G$ and $\ket{\psi}$ ranges over all unit vectors in $\mathbb{C}^n$.
    \end{definition}

    \begin{remark}
        There are two competing, dual conventions for the definition of orthonormal representations. Also, orthonormal representations are usually defined as vector arrangements in $\mathbb{R}^n$ rather than $\mathbb{C}^n$. Our definition was chosen because it gives a natural interpretation for orthonormal representations in terms of physics.

        Though our definition of orthonormal representations is non-standard, our definition of $\vartheta(G)$ is equivalent to the standard definition from \cite{Lovasz79}.
    \end{remark}
    
    Let $S_n$ be the orthogonality graph\footnote{meaning that $x\sim y \iff \left<x|y\right>=0$} of stabilizer vectors, identified up to phase. We will calculate $\alpha(S_n)$ and $\vartheta(S_n)$ to analyze the contextuality of the Pauli measurements. First, we develop an equivalent description of $S_n$ using the concepts developed in Section \ref{basics_section}.
    \begin{proposition}\label{pauli_measurements_orthogonalities_prop}
        Let $s,z\in V(S_n)$. Then $s\sim z$ iff $s$ and $z$ are inconsistent outcomes of two measurements $x,y\in \mathcal{L}^n_n$.
    \end{proposition}
    \begin{proof}
        If $s,z$ are inconsistent outcomes, then there is some $w\in \mathcal{L}^n$ such that $s$ and $z$ are in different eigenspaces of $w$, so $s\sim z$.

        Conversely, we use a standard fact about
          $\mathcal{P}(n)$. If $s\sim z \in S_n$ then there is some $w\in \mathcal{P}(n)$ such that $s$ and $z$ are in different eigenspaces of $w$ \cite{Aaronson_2004}. We extend $w$ to maximal Pauli measurements for $s$ and $z$ to obtain $x$ and $y$.
    \end{proof}
    
    The above proposition shows that $S_n$ is isomorphic to the graph of outcomes of measurements in $\mathcal{L}^n_n$ where adjacency is inconsistency in the sense of Definition \ref{consistency_def}. The next proposition expands on this equivalence. It shows that we can view partial hidden variable assignments of $\mathcal{L}^n$ as independent sets of $S_n$.
    
    \begin{proposition}\label{pval_independence_prop}
        There is a bijection between partial hidden variable assignments $f$ of $\mathcal{L}^n$ and independent sets $A \subset V(S_n)$. Under this bijection, $A \leftrightarrow Im(f)$. Hence, $\alpha(S_n)=Pval(\mathcal{L}^n)|\mathcal{L}^n_n|$.
    \end{proposition}
    \begin{proof}
        Let $f$ be a partial hidden variable assignment and let $A=Im(f)\subset V(S_n)$.

        Proposition \ref{pauli_measurements_orthogonalities_prop} shows that $A$ is indeed an independent set.
        
        Conversely, if $A\subset V(S_n)$ is an independent set, no two vertices of $A$ are outcomes of the same measurement in $\mathcal{L}^n_n$, because outcomes of a measurement are pairwise orthogonal. We define a partial hidden variable assignment for $\mathcal{L}^n$ by assigning measurements of $\mathcal{L}^n_n$ their unique outcome in $A$ if it exists and leaving the measurement unassigned otherwise. This partial hidden variable assignment cannot contain an inconsistency, else there would be two vertices in $A$ corresponding to vectors that lie in orthogonal eigenspaces of some measurement in $\mathcal{L}^n$. Since orthogonality defines adjacency of $S_n$, the inconsistency of the assignment contradicts that $A$ is an independent set. Therefore, the correspondence introduced in the first paragraph is indeed a bijection.
    \end{proof}

    \begin{proposition}\label{theta_count_prop}
        $\vartheta(S_n) = |\mathcal{L}^n_n|= \prod_{i=0}^{n-1} \frac{2^{2(n-i)}-1}{2^{n-i}-1}$.
    \end{proposition}
    \begin{proof}
    The stabilizer vectors are unit vectors in $\mathbb{C}^{2^n}$. The edges are defined by orthogonalities, so the stabilizer vectors are 
an orthonormal representation of $S_n$. Let $\ket{\psi}$ be any unit vector in $\mathbb{C}^{2^n}$. Our particular vectors are feasible for the program defining $\vartheta(S_n)$, and so define a lower bound as we now explain.
    
    We interpret the summands in Definition \ref{theta_definition} physically. The summand corresponding to $v\in V(S_n)$ is the probability that the measurement $\ket{f(v)}\bra{f(v)}$ gives outcome $1$ when measured on a system in state $\ket{\psi}\bra{\psi}$. Each stabilizer vector is an outcome of a unique measurement of $\mathcal{L}^n_n$, so $\mathcal{L}^n_n$ partitions $V(S_n)$. The outcomes of $x\in \mathcal{L}^n_n$ are $1$-dimensional subspaces, so we can represent the outcomes by units vector in the appropriate $1$-dimensional subspaces. Denote the outcome vectors for the measurement $x\in\mathcal{L}^n_n$ by $out(x)$. We group the summands according to this partition and obtain

    \[\vartheta(S_n)\geq \sum_{x\in \mathcal{L}^n_n} \sum_{\ket{s}\in out(x)} |\braket{\psi|s}|^2.\]
    
    For any measurement, the sum of probabilities of the outcomes is $1$, so this natural orthonormal representation gives a lower bound $|\mathcal{L}^n_n| \leq \vartheta(S_n)$.
    
    For the upper bound, observe that the sum of probabilities of outcomes for the measurements in $\mathcal{L}^n_n$ cannot exceed $1$, so a larger value for $\vartheta(S_n)$ is not possible. An equivalent way to see the upper bound is via the inequality from \cite{Knuth94}, $\vartheta(S_n)\leq \chi(\bar{S}_n)$.
    
    The formula for $|\mathcal{L}^n_k|$ is described in Lemma \ref{cval_approximation_helper_lemma}. Setting $k=n$ gives the second equality.
    
    \end{proof}
    
    \begin{remark}
        The orthonormal representation of stabilizer vectors is particularly nice in every unit vector $\ket{\psi}$ achieves the maximum in Definition \ref{theta_definition}. In other words, the maximal Pauli measurements exhibit state-independent contextuality in the sense of \cite{Cabello2015}.
    \end{remark}
    
    Using Theorem \ref{incompletness_thm} and Propositions \ref{theta_count_prop} and \ref{pval_independence_prop}, we easily obtain the next theorem.
    \begin{theorem}\label{stabilizer_alpha_theta_thm}
        \[\frac{\alpha(S_n)}{\vartheta(S_n)} \leq 2^{-n +O(1)}.\]
    \end{theorem}
    
    In \cite{Amaral_2015}, a natural measure of contextuality was introduced that also controls for the number of vertices of the graph. Their measure applies to families of graphs. We introduce an equivalent measure that applies to particular graphs.
    
    \begin{definition}
         Let $G$ be a graph. Define \[T(G)\coloneqq \frac{\log(\frac{\vartheta(G)}{\alpha(G)})}{\log(|V(G)|)}.\]
    \end{definition}

    Standard calculations show the following equivalence between $T$ and absolute maximal contextuality defined in \cite{Amaral_2015}.
    
    \begin{proposition}
        A family of graphs  $\{G_k\}_{k=1}^\infty$ exhibits maximal absolute contextuality in the sense of \cite{Amaral_2015} iff $\lim_{k\to \infty} T(G_k)= 1$.
    \end{proposition}
    Since $V(S_n) = 2^n |\mathcal{L}^n_n| \approx 2^{n^2}$, 
Theorem \ref{stabilizer_alpha_theta_thm} implies $\lim_{n\to \infty}T(S_n) = 0$ and the stabilizer graphs $S_n$ do not exhibit contextuality in the sense of \cite{Amaral_2015}, essentially because they have too many vertices. We describe a construction that rectifies this.

    It is easy to calculate that $T(G) = T(G\otimes G)$ for any graph $G$, where $\otimes$ is the disjunctive product. The $k$-fold disjunctive product of $G$ with itself is isomorphic to the orthogonality graph of outcomes of the Pauli measurements on $k$ separate systems. In other words, the corresponding physical scenario to $G^{\otimes k}$ has a list of $k$ Pauli measurements to be performed on $k$ separate systems as its measurements.

    Our construction is to sample the 
    tuples $(\mathcal{L}^n_n)^{k}_{i=1}$. Instead of all tuples, we use only the length-$k$ walks on a certain graph with vertex set $\mathcal{L}^n_n$. The effect of the sampling is to dramatically reduce the number of vertices while keeping the independence number small.
    
    \begin{theorem}\label{random_walk_construction_theorem}
    Let $d_n \coloneqq \lceil\frac{1}{(Pval(\mathcal{L}^n))^2}\rceil$. There exist explicit $d_n$-regular graphs $R_n$ with $V(R_n) \subset \mathcal{L}^n_n$ and $|V(R_n)| = |\mathcal{L}^n_n|(1-\epsilon_1)$, where $\epsilon_1 \leq \frac{1}{2}$. Moreover, the spectral parameter of $R_n$ satisfies $\lambda_n \leq 2\sqrt{d_n-1} + \epsilon_2$, where $\epsilon_2 \leq 1$. 
  
    Let $W_{n,k}$ be the orthogonality graph of outcomes of tuples of measurements defined by walks of length $k$ on $R_n$ as described above. Then 
	\[ \lim_{n\to \infty}\lim_{k\to \infty} T(W_{n,k})\geq \frac{1}{3}.\]
	\end{theorem}
    See Appendix \ref{supp_proofs_section} for proof. 

     The value of $\frac{1}{3}$ is competitive with examples in \cite{Amaral_2015}, though a graph with $T(G)=\frac{1}{2}$ is described in \cite{alo97}. Our example has the advantage of having an explicit orthonormal representation that is state-independent. Our construction is implementable with current technology, since Pauli measurements can be implemented with Clifford circuits and measurement in the computational basis.

   Any state-independent contextual scenario can be converted into nonlocality in a simple way: give Alice and Bob the maximally entangled state and allow them each the local measurements of the contextuality scenario \cite{cab21}. However, this procedure doesn't seem to preserve the strength of the violation of the noncontextual inequality as a local inequality. Theorem \ref{inconsistency_thm} allow us to derive a result about nonlocality as expressed by the following nonlocal game.

    \begin{definition}[Pauli Agreement Game]\label{pauli_agreement_game_def}
        Let $\mathcal{G}_n$ be the Pauli Agreement Game on $n$ qubits, defined as follows.
        
         First, the referee chooses $w\in \mathcal{L}^n_1$ and keeps it secret. He then chooses a pair of questions $(x,y) \in \mathcal{L}^n_n \times \mathcal{L}^n_n$ for Alice and Bob uniformly at random subject to $w\leq x, y$.
         
         Alice and Bob respond with outcomes for their respective measurements and win if their outcomes agree on $w$.
    \end{definition}
    \begin{theorem}There exists $\epsilon \in (0,1)$ such that $\frac{1}{2} \leq Val_{loc}(\mathcal{G}_n)\leq  \frac{1}{2} + \epsilon^n$.
    \end{theorem}
    \begin{proof}
        The lower bound is trivial because Alice and Bob can achieve a winning probability of $\frac{1}{2}$ by playing randomly.
    
        Fix Alice and Bob's strategies describe them with a pair of contextual hidden variable assignments for $\mathcal{L}^n_n$.
        
         Apply Theorem \ref{inconsistency_thm} to show that individually, Alice and Bob assign outcomes to $w$ with probabilities in the interval $(\frac{1}{2}-\delta,\frac{1}{2}+\delta)$ for $\delta \in O(\epsilon^n)$ for some $\epsilon \in (0,1)$. Based on this, their probability of winning $\mathcal{G}_n$ is at most $(\frac{1}{2}+\delta )^2 + (\frac{1}{2} - \delta)^2 = \frac{1}{2} +2\delta^2$. This shows that we can re-name $\epsilon$ such that $Val_{loc}(\mathcal{G}_n) \leq  \frac{1}{2} + \epsilon^n$.
    \end{proof}
    
    Of course $Val_{qtm}(\mathcal{G}_n) = 1$ because Alice and Bob can employ a strategy similar to the one in Proposition \ref{qtm_always_wins_prop}.

    Games like $\mathcal{G}_n$ where quantum strategies provide an advantage over local strategies can be mapped to violations of Bell inequalities so that the gap between the strategies maps to the strength of the Bell inequality violation \cite{Buhrman2012}. 
    
    In particular, our analysis of $\mathcal{G}_n$ maps to a family of Bell inequality violations \cite{Degorre09} that grow linearly in the number of qubits. This is not competitive with the best-known nonlocal games \cite{Buhrman2012}. The limiting factor in our analysis is our bound on $Val_{loc}(Z_{\frac{n}{2}})$ via parallel repetition. Our games do not provide larger Bell violations than can be obtained by parallel repetition.
\section{Discussion and Conclusions}\label{conclusions_section}
    The Pauli measurements have a natural combinatorial structure that we exploited to obtain non-trivial information about their behavior. In particular, we found incompleteness and inconsistency bounds (Theorem \ref{incompletness_thm} and Theorem \ref{inconsistency_thm}) for hidden variable assignments for the Pauli measurements. As applications, we produced graphs exhibiting contextuality in the sense of \cite{Amaral_2015} and showed that a natural nonlocal game has no local strategy that performs significantly better than the random strategy.

    The Pauli measurements are central to quantum information. The outcomes of $\mathcal{L}^n_k$ are stabilizer codes. They are a natural choice for the measurements needed for tomography. Many cryptographic protocols are based on the Pauli measurements. The eigenvalue techniques explored in this work may be versatile enough to apply to questions in these areas.

    In Section \ref{applications_section}, we described an explicit construction of graphs $W_{n,k}$ such that $\lim_{n\to \infty, k\to \infty}T(W_{n,k})\geq \frac{1}{3}$ that also exhibits state-independent contextuality. While it is known that there are families of graphs $\{G_k\}_{k=1}^\infty$ such that $\lim_{k\to \infty}T(G_k) = 1$ \cite{fei96}, it is not known whether the family still exists when the graphs are required to have optimal, state-independent orthonormal representations as in our case.
    
    We interpreted our bound on $Cval(\mathcal{L}^n)$ as a bound on the ability to win a certain natural nonlocal game using local resources. Using quantum resources, the games can be won with certainty. Such games give violations of Bell inequalities. The rate at which $Cval(\mathcal{L}^n)\to \frac{1}{4}$ determines the size of the violation. In our analysis, the limiting factor for this rate is given by our application of the parallel repetition lemma in Proposition \ref{parallel_prop}. Is there a better way to bound $Val_{loc}(Z_{\frac{n}{2}})$?
    
    It would be interesting to study the most natural Pauli game, whose questions are pairs from $\mathcal{L}^n_n\times \mathcal{L}^n_n$ that intersect in dimension $n-1$. As usual, the winning condition is that Alice and Bob must provide consistent answers on all questions in the intersection. We are not able to bound the value using eigenvalue techniques, due Proposition \ref{full_graph_spectrum_prop}.
    
    Our consistency games are similar to the games that define the quantum and classical chromatic numbers, $\chi_q$ and $\chi$ \cite{PAULSEN20162188}. In Proposition \ref{pval_independence_prop}, we interpret our incompleteness result as an upper bound on $\alpha(S_n)$, thus giving a lower bound on $\chi(S_n)$. We ask if it is possible to bound $\chi_q(S_n)$ by using the fact that our games can be won with certainty using a quantum strategy.
    
\section{Acknowledgements}
We thank Jitendra Prakash, Soumyadip Patra and Peter Beirhorst for some discussions. We thank Ad{\'{a}}n Cabello for a helpful comment regarding the introduction and providing useful references. We thank Andries Brouwer for helpful discussion regarding Appendix \ref{eigenvalue_conjectures_section}.

\printbibliography

\appendix

\section{Pauli Measurements}\label{Pauli_measurement_section}

Here, we justify identifying the Pauli measurements with the mathematical object $\mathcal{L}^n$ by showing exactly how $\mathcal{L}^n$ corresponds to commuting collections of elements of the Pauli group. Our construction is also described in \cite{Van_den_Nest_2004}.

First, let us define the Pauli group.

\begin{definition}

Let $X\coloneqq \begin{bmatrix} 0&& 1 \\ 1 && 0
\end{bmatrix}$ and $Z\coloneqq \begin{bmatrix}
                        1 && 0 \\
                        0 && -1
                    \end{bmatrix}$ be matrices acting on $\mathbb{C}^2$. For $x\in \mathbb{Z}_2^{2n}$ where $n\in \mathbb{N}$, define \[XZ(x) \coloneqq X^{x_1} Z^{x_{n+1}} \otimes X^{x_2}Z^{x_{n+2}} \otimes \dots \otimes X^{x_n} Z^{x_{2n}}.\]

We define the Pauli group on $n$ qubits to be \[\mathcal{P}(n)\coloneqq \{i^\omega XZ(x)\mid x \in \mathbb{Z}_2^{2n}, \omega \in \{0,1,2,3\}\}\] where the group operation is matrix multiplication.
\end{definition}

We will ignore phases and work with the vector space $\mathbb{Z}_2^{2n}$ instead by using the standard map $coord: \mathcal{P}(n)\to \mathbb{Z}_2^{2n}$ that sends $i^\omega XZ(x) \mapsto x$ for any $\omega \in \{0,1,2,3\}$. Algebraically, this map can be considered as quotienting the Pauli group by its center.

In the main text, we described $\mathcal{L}^n$ as the semilattice of isotropic subspaces of $\mathbb{Z}^{2n}_{2}$. Let us define the term ``isotropic.''
    
    \begin{definition}[Dot product]
			Let $a,b \in \mathbb {Z}_2^n$ with $a=(a_1,\dots,a_n)$ and $b=(b_1,\dots ,b_n)$. Then $a\cdot b \in \mathbb{Z}_2$ is defined by $a \cdot b \coloneqq \sum_{i=1}^n a_i b_i$, where operations are performed in $\mathbb{Z}_2$.
		\end{definition}
 		 
		\begin{definition}[Symplectic product]
			Let $a,b\in \mathbb{Z}^{2n}_2$ with $a=(a_1,a_2)$ and $b=(b_1,b_2)$ where $a_1,a_2$ are the projections of $a$ onto the first $n$ and second $n$ coordinates respectively, and similarly for $b_1,b_2$. Define the symplectic product $\left<a,b\right> \in \mathbb{Z}_2$ as being $0$ if $a_1 \cdot b_2 = a_2 \cdot b_1$ and $1$ otherwise.
	 	\end{definition}

\begin{definition}[Isotropic]
    		 	A subspace $M$ of $\mathbb{Z}_2^{2n}$ is called isotropic if $\left<x,y\right>=0,\forall x,y\in M$.
\end{definition}
When we move from $\mathcal{P}(n)$ to $\mathbb{Z}_2^{2n}$ via $coord$, we need to keep track of the commutation relations that define comeasurability of the observables. This is achieved with the following lemma.

	 \begin{lemma}\cite{Hostens04}
	 	$i^{\omega_1} XZ(a), i^{\omega_2} XZ(b)\in \mathcal{P}(n)$ commute iff $\left<a,b\right>=0.$
 	\end{lemma}
 	\begin{proof}
 		Since phase factors do not affect commutation, we assume that $A= XZ(a)$ and $B = XZ(b)$, and let $a=(a_1,a_2)$ and $b=(b_1,b_2)$ as in the definition of the symplectic product.
 		
 		The commutation relation $ZX=XZ$ implies that $AB = (-1)^\omega XZ(a+b)$ where $\omega = b_1\cdot a_2 $. So $A$ and $B$ commute exactly when $AB= (-1)^{\omega_1} XZ(x + y) = BA  = (-1)^{\omega_2} XZ(a+b)$, which occurs exactly when $\omega_1 = \omega_2$, or equivalently $b_1 \cdot a_2 = a_1 \cdot b_2$ which by definition means $\left<a,b\right>=0$.
 	\end{proof}
	 	
The lemma explains why the ``isotropic'' condition appears in Definition \ref{Pauli_measurement_def}. We justify the name ``Pauli measurement'' by the fact that the pre-image of an isotropic subspace under $coord$ is a commutative subset of $\mathcal{P}(n)$, and so can be performed as a measurement in the standard quantum formalism.

To describe the relationship between $\mathcal{L}^n$ and its usual formulation in quantum mechanics further, we describe the outcomes in both cases by defining a partial inverse to $coord$.

  \begin{definition}[Conventional phase]
  	 The conventional phase is defined by $\widetilde{XZ}(x) \coloneqq i^{x_1\cdot x_2} XZ(x)\in \mathcal{P}(n)$ where $x=(x_1,x_2)\in \mathbb{Z}^{2n}_2$ and $x_1,x_2 \in \mathbb{Z}_2^n$ are the projections of $x$ onto the coordinates with indices $\{1,\dots,n\}$ and $\{n+1,\dots,2n\}$ respectively.
  \end{definition}

In the usual formulation of quantum mechanics, the outcomes of a Hermitian element of $\mathcal{P}(n)$ are its eigenvalues, $\pm 1$. We can define the outcome of a vector in $\mathbb{Z}_2^{2n}$ by using the conventional phase to convert it to an element of $\mathcal{P}(n)$ and describing the outcome there. We also switch to an additive notation: if we denote the outcomes by $f$, then $f(\widetilde{XZ}(x)) = +1$, equivalently expresses that $f(x)=0$ and $f(\widetilde{XZ}(x)) = -1$ equivalently expresses that $f(x)=1$.

Let $A,B$ be commuting Hermitian operators. The functional composition principle \cite{sep-kochen-specker} describes how $f(AB)$ can be deduced from $f(A)$ and $f(B)$. By converting vectors to Pauli operators using the conventional phase and applying the functional composition principle, we derive the following constraint on outcome assignments in terms of vectors in $\mathbb{Z}^{2n}_2$.

\begin{proposition}(Linearity/antilinearity constraints)

	 	\label{lin_antilin}
	 		If $x, y \in \mathbb{Z}^{2n}_2$ satisfy $\left<x,y\right> = 0$, any outcome $f$ that assigns values to $\widetilde{XZ}(x)$ and $\widetilde{XZ}(y)$ also assigns a value to $\widetilde{XZ}(x+y)$ according to $f(x+y)= w + f(x)+f(y)$ where addition on the right of the equality is performed in $\mathbb{Z}_2$ and $w=0$ if $\frac{i^{x_1\cdot x_2} i^{y_1 \cdot y_2}}{i^{(x_1+y_1) \cdot (x_2 +y _2)}}(-1)^{x_2 \cdot y_1}=1$ and $w=1$ otherwise.
	 	\end{proposition}
   
If $S\in \mathcal{L}^n$ , we call a function $f:S\to \mathbb{Z}_2$ that satisfies the condition of Proposition \ref{lin_antilin} a linear/antilinear function. The outcomes of $S$ defined to be the outcomes of $\{\widetilde{XZ}(x) |x \in S\}$. The following proposition shows that 
 the outcomes and linear/antilinear functions are equivalent, so linear/antilinear functions provide a concrete notation for describing the outcomes of the Pauli measurements.

\begin{proposition}\label{pauli_outcome_correspondence_prop}
         Let $S \in \mathcal{L}^n$ be a Pauli measurement. Let $B\subset S$ be a basis for $S$ over $\mathbb{Z}_2$. Then the outcomes of the  $S$ correspond bijectively to the $2^{|B|}$ assignments of $B\to \mathbb{Z}_2$, and these are in bijective correspondence with linear/antilinear functions on $S \to \mathbb{Z}_2$.
\end{proposition}

\begin{proof}
    Every outcome of a Pauli measurement induces a unique linear/antilinear function as in Proposition \ref{lin_antilin}. Each linear/antilinear function restricts to a unique assignment of $B\to \mathbb{Z}_2$. 
    
    There are $2^{|B|}$ different outcomes for the measurement, since the measurement consists of performing $|B|$ $2$-outcome measurements simultaneously. There are also $2^{|B|}$ assignments $B \to \mathbb{Z}_2$. This counting shows that the associations in the first paragraph are bijective. Hence, every assignment $B\to \mathbb{Z}_2$ can be extended to a unique linear/antilinear function $S\to \mathbb{Z}_2$.
\end{proof}

\begin{corollary}\label{lin_cor}
    Let $S\subset \mathbb{Z}_2^{2n}$ be a Pauli measurement spanned by vectors of the form $(v,0)$ where $v\in \mathbb{Z}_2^n$, and by $0$ we mean the zero vector of $\mathbb{Z}^n_2$. Then the outcomes of $S$ are linear functions $S\to\mathbb{Z}_2$.
\end{corollary}

\begin{proof}
    Our linearity/antilinearity conditions described in Proposition \ref{lin_antilin} are always linear on $S$, so the outcomes are linear functions.
\end{proof}
 
See \cite{calderbank97} for some discussion of the algebraic structure of spaces whose outcomes are described by linear functions.

The symmetry group of the Pauli measurements $Sp(2n,\mathbb{Z}_2)$ \cite{Put} \cite{Tolar_2018} acts transitively on isotropic subspaces of a given dimension. Therefore, Corollary \ref{lin_cor} shows that without loss of generality, any particular Pauli measurement has outcomes that can be described by linear functions. This gives us the following corollary that will be useful for our inconsistency result.

\begin{corollary}\label{disagreement_on_half_cor}
    If $S_1,S_2 \in \mathcal{L}^n$ are Pauli measurements and $f_1:S_1\to \mathbb{Z}_2$ and $f_2:S_2\to \mathbb{Z}_2$ are inconsistent outcomes, then  $|\{ w\in S_1 \cap S_2 | f_1(w)\neq f_2(w)\}| = 2^{\dim(S_1\cap S_2)-1}$.
\end{corollary}
\begin{proof}
    Without loss of generality (appealing to Corollary \ref{lin_cor}), assume elements of $S_1$ have the form $(v,0)$ so that $f_1$ is a linear function and $f_2$ restricted to $S_1 \cap S_2$ is also a linear function. Then $f_1-f_2$ is a non-zero linear function from $S_1 \cap S_2$ to $\mathbb{Z}_2$. The rank-nullity theorem implies that $\dim(\ker((f_1 -f_2)|_{S_1\cap S_2})) = \dim(S_1\cap S_2)
    -1$, and this completes the proof.
\end{proof}
\section{Supplemental Proofs}\label{supp_proofs_section}
This section contains some supplemental proofs.

\begin{proof} (Of Proposition \ref{L22_lemma_prop}) Define a square subspace $S \subset \mathbb{Z}^{2n}_2$ to be an $n-1$-dimensional subspace of $\mathbb{Z}_2^{2n}$ with a symplectic basis- a basis $\{x_1, x_2,\dots x_{n-1}\} \cup \{z_1,z_2,\dots z_{n-1}\}$ such that $\left<x_i,z_i\right> = 1$ for all $i$ and all other symplectic inner products between pairs of basis elements are $0$. 

It is easy to see that $Sp(2n,\mathbb{Z}_2)$ acts transitively on square subspaces and that each isotropic $2$-dimensional subspace is contained in the same number of square subspaces.

To choose a random element of $\mathcal{L}^n_2$, we first choose a random square subspace $S$, then choose an isotropic $2$-dimensional subspace of $S$.

If we assign outcomes to more than $\frac{4}{5}$ of the elements in $\mathcal{L}^n_2$, then we must also assign outcomes to more than $\frac{4}{5}$ of the measurements in some square subspace, isomorphic to $\mathcal{L}^{n-1}_2$.

This is impossible because
square subspaces can be thought of as a sub-system of $(n-1)$ qubits \cite{Zanardi_2004}. In other words,
the semilattice of isotropic subspaces of a square subspaces is isomorphic to $\mathcal{L}^{n-1}$ so we have a bound by induction. The base case is $Pval(\mathcal{L}^2_2)= \frac{12}{15}$, as found by a direct computation.
\end{proof}

\begin{proof}
(Of Proposition \ref{partial_var_prop}) Suppose that $f$ is a partial hidden variable assignment for $\mathcal{L}^n$ with domain $B\subset \mathcal{L}^n_n \subset V(B_{n,2})$ and let $\beta \coloneqq \frac{|B|}{|\mathcal{L}^n_2|}$. 

Then $f$ restricts to a partial hidden variable assignment for $\downarrow\mathcal{L}^n_2$ that is defined for any $a\in \mathcal{L}^n_2$ such that $a\leq b$ for some $b\in B$. Call this set $A$ and its complement $\bar{A} \coloneqq \mathcal{L}^n_2 - A$. Define $c \coloneqq \frac{|A|}{|\mathcal{L}^n_2|}$ and $d\coloneqq \Delta(B_{n,2})$

Lemma \ref{biregular_expander_mixing_lemma} states that

\[
|\frac{E(\bar{A},B)}{E(B_{n,2})} - (1-c) \beta| \leq \frac{\lambda}{d}\sqrt{(1-c) c (1-\beta) \beta}.
\]

Note that $E(\bar{A},B) =0$ by definition. Squaring both sides and some simple algebra gives 
\[
\beta \leq \frac{\lambda^2 c}{d^2 (1-c)}.
\]

Since $\frac{4}{5}\geq c$, we find that $\beta \leq 4(\frac{\lambda}{d})^2$. This argument applies to the domain of any partial hidden variable model $f$ and so bounds $Pval(\mathcal{L}^n_n)$ above.
\end{proof}

\begin{proof}
    (Of Proposition \ref{qtm_always_wins_prop})

    To show that Alice and Bob's strategy works, we must show that for $z=(z_1, z_2)\in x\cap y$, Alice and Bob's outcomes agree when $z_1 \cdot z_2=0$ and disagree when $z_1\cdot z_2=1$. To see this, let $c=(0,0,0,1)\in (\mathbb{Z}^2_2)^2$. Let $U$ be a Clifford operator such that $U^\dagger \widetilde{XZ}(z) U = \widetilde{XZ}(c)$, and note that the measurement operator $\widetilde{XZ}(c)$ corresponds to measuring the second qubit in the computational basis. Let $\overline{U}$ be the complex conjugate of $U$ and let $\tau =\begin{cases}1 &\text{if } z_1\cdot z_2 =0\\ -1 &\text{otherwise}\end{cases}$. Because $XZ(z)$ has real entries, it follows that $\overline{\widetilde{XZ}(z)} = \tau \widetilde{XZ}(z)$, and $\overline{\widetilde{XZ}(c)} = \widetilde{XZ}(c)$. From these two facts, we find that $\tau \overline{U^\dagger} \widetilde{XZ}(z) \overline{U}=\overline{U^\dagger \widetilde{XZ}(z) U} =\tau \widetilde{XZ}(c)$.
        
        We use the well-known fact that $V\otimes \overline{V}\ket{\psi} = \ket{\psi}$, for any unitary operator $V$ of appropriate dimensions (in particular, $U$). 
        
        \[Tr( \widetilde{XZ}(z) \otimes  \widetilde{XZ}(z) \ket{\psi}\bra{\psi}) = Tr(\widetilde{XZ}(z) \otimes  \widetilde{XZ}(z) (U\otimes \overline{U})\ket{\psi}\bra{\psi} (U \otimes \overline{U})^\dagger). \]

        Using the cyclic property of the trace and combining tensor factors, we find

        \[
        =Tr( U^\dagger \widetilde{XZ}(z) U \otimes \overline{U}^\dagger \widetilde{XZ}(z)\overline{U} \ket{\psi}\bra{\psi}).
        \]

        \[
        = Tr( \widetilde{XZ}(c) \otimes \tau \widetilde{XZ}(c) \ket{\psi}\bra{\psi})= \tau Tr(\ket{\psi}\bra{\psi})=\tau.
        \]
        In the last line, we have used the fact $V\otimes \overline{V} \ket{\psi}\bra{\psi}$ with $V=XZ(c)$.

        The calculation shows that when $z_1\cdot z_2=0$ (and hence $\tau = 1$), the expected value of the measurement $\widetilde{XZ}(z) \otimes  \widetilde{XZ}(z)$ is $1$. Since $1$ is also the maximum eigenvalue of the operator, the outcome corresponding to eigenvalue $1$ occurs with certainty. Therefore, when Alice and Bob perform the measurements $\widetilde{XZ}(z)\otimes I$ and $I \otimes \widetilde{XZ}(z)$, the products of the eigenvalues corresponding to their outcomes are $1$, and therefore their outcomes agree.

        Similarly, when $z_1\cdot z_2=1$ (and hence $\tau = -1$), the expected value of the measurement $\widetilde{XZ}(z) \otimes  \widetilde{XZ}(z)$ is $-1$, which is also the smallest eigenvalue of the operator, so it must occur with certainty. This shows that the outcomes of the measurements $\widetilde{XZ}(z)\otimes I$ and $I \otimes \widetilde{XZ}(z)$ have product $-1$, and therefore disagree with certainty.
\end{proof}

\begin{proof}
(Of Theorem \ref{random_walk_construction_theorem}) The graphs $R_n$ exist by the construction in \cite{mohanty19}.

Let $A\subset V(W_{n,k})$ be a maximal independent set. Let $M$ be the set of measurements represented by length-$k$ walks on $R_n$. Let $M^\prime\subset M$ be the set of measurements whose outcomes intersect $A$. Since different outcomes of the same measurement are orthogonal,  and because each each outcome corresponds to at most $1$ measurement, we have $|A|=|M^\prime|$.

We find that $|M| = \vartheta(W_{n,k})$ by appealing to an argument similar to the proof of Proposition \ref{theta_count_prop}. Therefore, $\frac{|M^\prime|}{|M|} = \frac{\alpha(W_{n,k})}{\vartheta({W_{n,k}})}$. We calculate the probability that a randomly chosen measurement of $M$ is in $M^\prime$.

Let $x,y\in A$. We can write $x=(x_i)_{i=1}^k$ and $y=(y_i)_{i=1}^k$ where $x_i$ and $y_i$ are outcomes of measurements in $\mathcal{L}^n_n$. Because $x\not\sim y$, $x_i \not \sim y_i$ for all $1\leq i \leq k$. Hence, $A$ projects to an independent set $A_i\subset S_n$ on the $i^{th}$ coordinate. 

Let $M^\prime_i \subset \mathcal{L}^n_n$ be the set of measurements for the $i^{th}$ system that have an outcome in $A_i$. Equivalently, $M_i^\prime$ is the set of $i^{th}$ coordinates of measurements in $M^\prime$. A length-$k$ walk $v$ in $R_n$ is in $M^\prime$ iff it is in $M^\prime_i$ at each step, so we can apply Lemma \ref{hitting_lemma}.
\[
\frac{\alpha(W_{n,k})}{\vartheta(W_{n,k})} = Pr(v\in A)\leq \left(\frac{Pval(\mathcal{L}^n)}{1-\epsilon_1}  + \frac{\lambda_n}{d_n}(1-\frac{Pval(\mathcal{L}^n)}{1-\epsilon_1})\right)^k.
\]

Our assumptions (and some tedious algebra) about $\lambda_n$ and $d_n$ imply that $\frac{\lambda_n}{d_n} \leq 3 Pval(\mathcal{L}^n)$ when $n$ is large enough. We obtain the inequality

\[\frac{\alpha(W_{n,k})}{\vartheta (W_{n,k})} \leq (5 Pval(\mathcal{L}^n) )^k. \]

We count the total number of length-$k$ walks of $R_n$ by starting at any vertex, then taking $k-1$ steps, each with $d_n$ choices. Our choice of $d_n =\frac{1}{\lceil Pval(\mathcal{L}^n)^2 \rceil}$ allows the close approximation $d^{k-1} \approx Pval(\mathcal{L}^n)^{2(1-k)}$ for large $k$. \[|V(W_{n,k})| = \frac{|\mathcal{L}^n_n|}{1-\epsilon_1} d^{k-1}_n 2^{nk} \approx \frac{1}{1-\epsilon_1} |\mathcal{L}^n_n| Pval(\mathcal{L}^n)^{2(1-k)} 2^{nk}.\]

This implies that 
\[
T(W_{n,k})\geq -\frac{k\log(5 Pval(\mathcal{L}^n))}{\log(|\mathcal{L}^n_n|) + 2(1-k)\log(Pval(\mathcal{L}^n)) +nk - (1-\epsilon_1)}.
\]

Taking the limit as $k\to \infty$ while holding $n$ fixed gives

\[\lim_{k\to \infty} T(W_{n,k}) \geq - \frac{\log(5Pval(\mathcal{L}^n))}{-2\log(Pval(\mathcal{L}^n)) +n}.\]

We apply Theorem \ref{incomplete_inconsistent_cor} and calculate the limit $n\to \infty$ to arrive at the final result,

\[ \lim_{n\to \infty } \lim_{k\to \infty} T(W_{n,k}) \geq \frac{1}{3}.\]
\end{proof}
\section{Eigenvalue Propositions}\label{eigenvalue_conjectures_section}

In this section, we prove Propositions \ref{full_graph_spectrum_prop}, \ref{half_graph_spectrum_prop}, and \ref{bipartite_spectrum_prop} by relating the graphs $G_w^\prime, G_w$ and $B_{n,2}$ to dual-polar graphs, defined in \cite{Bro89} and analyzed further in \cite{Bro18}.

\begin{definition}[Dual-polar graphs]
    For $n\in \mathbb{N}$, $i\geq 1$, define $C^i_n(2)$ to be the graph with vertex set $V(C^i_n(2)) \coloneqq \mathcal{L}^n_n$, and edges defined by $x\sim y \iff d(x,y)=i$.
\end{definition}

\begin{proposition}\label{lattice_isomorphism_prop}
    Let  $w\in \mathcal{L}^n_1$. Then $\mathcal{L}^n \cap \uparrow w \simeq \mathcal{L}^{n-1}$.
\end{proposition}
\begin{proof}

    We will show that $w^\perp / w$ is isomorphic to $\mathbb{Z}^{2(n-1)}_2$, where $w^\perp \coloneqq \{v \in \mathbb{Z}^{2n}_2 | \left<v,w\right>=0\}$. The isomorphism preserves the symplectic product, and therefore  induces an isomorphism between the two semilattices $\mathcal{L}^n \cap \uparrow w$ and $\mathcal{L}^{n-1}$.
    Without loss of generality, we can assume that $e_1$ is the unique vector in $w$ and the first element of a basis $\{e_i\}_{i=1}^{2n}$ for $\mathbb{Z}_2^{2n}$ satisfying the standard symplectic relations, $\left< e_i,e_j  \right> = 0$ unless $i=j \text{ mod }n$ in which case $\left< e_i,e_j  \right> = 1$.

    Using the definition of the symplectic product, we can see that $w^\perp = span(\{e_i\}_{i\neq n+1})$, so $w^\perp/w$ has a basis $\{e_i+w\}_{i\neq 1, n+1}$. This shows that $w^\perp / w \simeq \mathbb{Z}^{2n}_2$ as a linear space. To see that the isomorphism preserves the symplectic product, let $a,b \in w^\perp$. Then linearity of the symplectic product gives $\left<e_i+w,e_j\right> = \left<e_i,e_j\right>+\left<w,e_j\right>=\left<e_i,e_j\right>$, so the symplectic product is constant on $w$-cosets, and therefore well-defined on $w^\perp / w$. Hence, $\{e_i+w\}_{i\neq 1,n+1}$ satisfies the standard commutation relations as a basis of $\mathbb{Z}^{2(n-1)}_2$ and the isomorphism preserves the symplectic product. This induces an isomorphism on the semilattices as claimed.
\end{proof}

\begin{proposition}\label{dual_polar_graph_prop}
    $G_w^\prime \simeq C^1_{n-1}(2)$ and
    $G_w \simeq C_{n-1}^{\frac{n}{2}}(2)$.
\end{proposition}
\begin{proof}
    The semilattice isomorphism of Proposition \ref{lattice_isomorphism_prop} induces a bijection on the top-level. Since distance is defined in terms of the semilattice, it is preserved by the bijection.
\end{proof}

The eigenvalues of $C_n^i(2)$ have been studied in the literature, and we reference these known values to prove Proposition \ref{full_graph_spectrum_prop}.

\begin{proof}(Of Proposition \ref{full_graph_spectrum_prop})
    The eigenvalues of $C^1_n(2)$ are $\{2{n-k\choose 1}_2-{k\choose 1}_2 \mid k\in\{0,\dots,n\} \}$ \cite{Bro18}. Clearly, the eigenvalues are ordered with $k=0$ being the largest. This gives $\Delta(C_n^1(2)) = 2{n\choose 1}$.

    The second-largest eigenvalue in absolute value either corresponds to $k=1$ or $k=n$, and direct calculation reveals that $k=1$ has the greater magnitude. 
    
    \[\lambda = \max(|2{n-1 \choose 1}_2 - 1|, |2 - {n \choose 1}_2|) = 2^{n-1}-2.\]

    The spectral ratio is
    \[ R(C_n^1(2)) =  \frac{2^{n-1}-2}{2{n\choose 1}_2} = \frac{2^{n-1} -2}{2^n-1} = \frac{1}{2} - \frac{1}{2^n-1}.\]

    Taking the logarithm of both sides and replacing $n$ with $n-1$ gives 
    \[
    \log(R(G_w^\prime)) = \log(R(C^1_{n-1}(2))) \approx -1 - \frac{2}{2^{n-1}-1}
    \]
    using the standard approximation $\log(1+x) \approx x$ for small $x$.
\end{proof}

\begin{proof}(Of Proposition \ref{half_graph_spectrum_prop})
    We use Proposition \ref{dual_polar_graph_prop} and reference \cite{Bro18} to find formulas for the eigenvalues of $C_{n-1}^{\frac{n}{2}}(2)$.

    The largest eigenvalue is $\Delta(C^{\frac{n}{2}}_{n-1}) = {n-1\choose\frac{n}{2} }_2 2^{{\frac{n}{2}+1 \choose 2}}$, according to part $3$ of Lemma \ref{cval_approximation_helper_lemma}.

    Proposition $4.1$ of \cite{Bro18} shows that $C_{n-1}^{\frac{n}{2}}(2)$ has an eigenvalue $(-1)^{\frac{n}{2}} {n-1 \choose \frac{n}{2} } 2^{\frac{n}{2} \choose 2} $, while Proposition $6.4$ part (iii) of \cite{Bro18} shows that the absolute value of this eigenvalue is the spectral parameter.

    Hence, 
    \[R(G_w)=R(C_{n-1}^{\frac{n}{2}}(2)) = \frac{{n-1\choose\frac{n}{2} }_2 2^{{\frac{n}{2} \choose 2}}}{{n-1 \choose \frac{n}{2} } 2^{\frac{n}{2} +1 \choose 2}} = 2^{{\frac{n}{2} \choose 2}-{\frac{n}{2} +1 \choose 2} } = 2^{-\frac{n}{2}}.\]

    Taking the log of both sides gives the result.
\end{proof}

Proposition \ref{bipartite_spectrum_prop} is more involved and we need some lemmas.

\begin{lemma}\label{actual_mediant_inequality}
    Let $n,m\in \mathbb{N}$ and assume $m\leq \frac{(n-1)}{2}$. Then ${n\choose m}_2 \geq 2^{m(n-m)}$.
    \end{lemma}
    \begin{proof}
        By definition, 
        \[{n\choose m}_2= \prod_{i=0}^{m-1}\frac{2^{n-i}-1}{2^{i+1}-1}.\]

        Our assumption on $n$ and $m$ implies that each factor is greater than $1$, so we can apply the mediant inequality to each factor to conclude that 

        \[{n \choose m}_2 \geq \prod_{i=0}^{m-1} \frac{2^{n-i}}{2^{i+1}} = 2^{\sum_{i=0}^{m-1} n-2i-1} = 2^{m(n-m)}.\]
    \end{proof}

\begin{lemma}\label{mediant_lemma}
    Let $x,y \in \mathbb{R}_+$. Then $\frac{x}{y}\leq \frac{x+1}{y+1}\frac{1}{1-\frac{1}{y+1}}$.
    \end{lemma}
    \begin{proof}
        This is just simple algebra.
        \[
        \frac{x}{y}- \frac{x+1}{y+1} = \frac{x-y}{y(y+1)} \leq \frac{x}{y(y+1)}.
        \]

        Re-arranging gives the result.

        \[
        \frac{x}{y}-\frac{x}{y} \left( \frac{1}{y+1}\right) = \frac{x}{y}\left( 1 - \frac{1}{y+1}\right) \leq \frac{x+1}{y+1}
        \]
    \end{proof}

    The following lemma is proven in \cite{product_reference} using advanced number theoretic techniques.
    
    \begin{lemma}\label{product_lemma}
    $\prod_{i=1}^\infty \frac{1}{(1-\frac{1}{2^i})} < 5$.
    \end{lemma}
    \begin{corollary}\label{gauss_estimation_cor}
        Let $n,m\in \mathbb{N}$. Then
        ${n\choose m}_2 \leq 5\cdot 2^{m(n-m)}$.
    \end{corollary}
    \begin{proof}
        By definition, 
        
        \[{n\choose m}_2= \prod_{i=0}^{m-1}\frac{2^{n-i}-1}{2^{i+1}-1}.\]

        We use Lemma \ref{mediant_lemma} on each factor with $x=2^{n-i}-1$ and $y = 2^{i+1}-1$ to obtain
        \[{n\choose m}_2 \leq \prod_{i=0}^{m-1} 2^{n-2i-1} \prod_{i=0}^{m-1}\frac{1}{(1-\frac{1}{2^{i+1}})}.
        \]

        We apply Lemma \ref{product_lemma} to the right product. We turn the left product into a sum on the exponent to find 

        \[{n \choose m}_2 \leq 5 \cdot 2^{\sum_{i=0}^{m-1} n-2i-1} = 5 \cdot 2^{m(n-m)}.\]
    \end{proof}

  \begin{lemma}\label{summation_lemma}
        Let $n>3$ be a natural number. Let 

        \[Q(i) \coloneqq 2(n-i-2) + i(n-i) + {i\choose 2}. \]

        Then 
        \[\sum_{i=0}^{n-2} 2^{Q(i)}  \leq 3 \cdot 2^{Q(n-3)}.\]
    \end{lemma}
    \begin{proof}
        We can expand the quadratic equation into

        \[
        Q(i) = - \frac{1}{2} i^2 + (n-\frac{5}{2})i + (2n-4).
        \]
        To find the extrema, we find the value of $i$ for which $Q^\prime(i)=0$.

        \[Q^\prime(i) = -i +n-\frac{5}{2} =0.\]

        The two nearest integers to $n-\frac{5}{2}$ are $i=n-2$ or $i=n-3$, so $Q(n-2)=Q(n-3)$ are the maximal values for $i \in \mathbb{N}$.

        Because $Q(i)$ is quadratic with negative $i^2$ coefficient, the sequence $(Q(0),Q(1),\dots, Q(n-3))$ is a monotone strictly increasing sequence of integers. Hence,
        \[
        \sum_{i=0}^{n-2} 2^{Q(i)} \leq 2^{Q(n-2)} + \sum_{k=0}^{Q(n-3)} 2^k < 3 \cdot 2^{Q(n-3)} 
        \]
        where the last inequality uses the formula for the sum of a geometric sequence.
    \end{proof}

    To work towards a proof of Proposition \ref{bipartite_spectrum_prop}, let $X$ be the adjacency matrix of $B_{n,2}$, and let $\lambda$ and $\Delta$ be the spectral parameter and $\max(Spec(B_{n,2}))$ respectively.

   Because $B_{n,2}$ is bipartite, we can order the vertices such that $X$ is block-antidiagonal, and has the form $X=\begin{bmatrix}
        0 && B\\
        B^T && 0
        \end{bmatrix}$.

        It follows that \[X^2 = \begin{bmatrix}
        BB^T && 0\\
        0 && B^TB
        \end{bmatrix}.\]

        The nonzero eigenvalues of $BB^T$ and $B^TB$ are the same (with different multiplicities).

        Thus, $\Delta$ and $\lambda$ can be computed as the largest and second-largest eigenvalues of $BB^T$.

        The entries of $B B^T$ (indexed by $(v,w)$) count the number length-2 walks $v\sim t \sim w$, where $\dim(v)=\dim(w)=n$ and $\dim(t)=2$.

        This gives us the decomposition
        \[
        B B^T = \sum_{i=0}^n {n-i\choose 2}_2 A_i
        \]

        where $A_i$ is the adjacency matrix of $C^i_{n}(2)$. The decomposition shows that $B B^T$ is in the Bose-Mesner algebra \cite{Bro89} so its eigenspaces are the idempotents, $E_i$, of the algebra.

        We use the eigendecompositions for $A_i$, given by
        $A_j= \sum_{i} P_{i,j} E_i$, where $P_{i,j}$ is the $i^{th}$ eigenvalue\footnote{ordered by the corresponding eigenvalue for $C^1_n(2)$} of $C^j_{n}(2)$ along with the fact that the $E_i$ are eigenspaces of $BB^T$ to conclude that the eigenvalues of $BB^T$ are $\{\sum_{i=0}^n {n-i \choose 2}_2 P_{h,i} \mid h \in \{0,\dots,n\}\}$. 
        
        Since $E_0$ is the idempotent corresponding to the largest eigenvalue of each $A_i$, $E_0$ is also the idempotent for the largest eigenvalue of $BB^T$.

    We have established
    \begin{proposition}\label{background_proposition}
        The eigenvalues of $X^2$ are given by
        
        \[Spec(X^2) = \{\sum_{i=0}^n {n-i \choose 2}_2 P_{h,i} \mid h \in\{0,\dots,n\}\}.
        \]

        The largest eigenvalue is given by setting $h=0$.
    \end{proposition}

    Armed with expressions for the spectrum of $X^2$, we find upper bounds for $\lambda^2$ and lower bounds for $\Delta^2$.

    \begin{proposition}\label{lambda_upper_bound}
        $\lambda^2 \leq 75 \cdot 2^{Q(n-3)}$, where $Q(i)=2(n-i-2) + i(n-i) + {i \choose 2}.$
    \end{proposition}
    \begin{proof}
        We use Proposition \ref{background_proposition}, along with the fact that $h=0$ corresponds to the largest eigenvalue to conclude that
        \[\lambda^2 = \max_{h\neq 0}( | \sum_{i=0}^n {n-i\choose 2}_2 P_{h,i} |)
        \leq  \max_{h\neq 0} (\sum_{i=0}^n {n-i\choose 2}_2 |P_{h,i}|) = \sum_{i=0}^n {n-i\choose 2}_2 {n \choose i}_2 2^{{i\choose 2}}.
        \]

    In the second inequality, we use the formula for the spectral parameter of $C^i_n(2)$ known from \cite{Bro18} that we used in the proof of Proposition \ref{half_graph_spectrum_prop}.

        We apply Corollary \ref{gauss_estimation_cor} to obtain

        \[\lambda^2 \leq 25\sum_{i=0}^{n-2} 2^{2(n-i-2) + i(n-i) + {i \choose 2}}.\]

        Note that $Q(i)$ is the exponent.
        Using Lemma \ref{summation_lemma}, we find

        \[\lambda^2 \leq 75 \cdot 2^{Q(n-3)}.\]
    \end{proof}
\begin{proposition}\label{d_lower_bound}
        For $n> 8$, $\Delta^2\geq 2^{3(n-3) + {n-2\choose 2}}$.
    \end{proposition}
    \begin{proof}
        $\Delta^2$ is the largest eigenvalue of $BB^T$, and it is the eigenvalue associated with the idempotent $E_0$. Thus, we have an explicit formula for $\Delta^2$ by Proposition \ref{background_proposition}.

        \[\Delta^2 = \sum_{i=0}^n {n-i \choose 2}_2 P_{0,i}\]

        where $P_{0,i}$ is the largest eigenvalue of $A_{i}$ and is given by the formula

        \[
        P_{0,i} = {n \choose i}_2 2^{{i+1 \choose 2}}.
        \]

        All terms in the sum are non-negative. We keep only the $i=n-3$ term and ignore the others. \[\Delta^2 \geq {3 \choose 2}_2 {n\choose 3}_2 2^{{n-2 \choose 2}} \geq 2^{3(n-3) + {n-2\choose 2}}.\] The second inequality is Lemma \ref{actual_mediant_inequality}, and our use of this lemma explains why we assume $n>8$.
    \end{proof}

    \begin{proof}(Of Proposition \ref{bipartite_spectrum_prop})
        We use the bounds of Propositions \ref{lambda_upper_bound} and \ref{d_lower_bound} to conclude that for some constant $C$ independent of $n$,

        \[
        \frac{\lambda^2}{\Delta^2} \leq \frac{C \cdot 2^{3n +{n-3 \choose 2}}}{2^{3n+{n-2 \choose 2}}} = C2^{{n-3 \choose 2} - {n-2 \choose 2} } = C 2^{-(n-3)}.
        \]

    Thus, up to a different constant $C^\prime$, 

    \[
    \frac{\lambda}{\Delta} \leq C^\prime 2^{-\frac{n}{2}}
    \]
    and taking the log of both sides gives the result.
    \end{proof}

\section{Solutions to Counting Problems}\label{degree_count_section}
In this section, we derive formulas for $\mathcal{L}^n_k$, $Q$, and $\Delta(G_w)$.

These objects are defined in terms of isotropic subspaces, so we count them with methods similar to the q-analog from combinatorics, defined in \cite{van_lint_wilson_2001}.

\begin{definition}\label{q_analog_def}
     For $n,m\in \mathbb{N}$, define ${n \choose m}_2 \in \mathbb{N}$ to be the number of $m$-dimensional subspaces of $\mathbb{Z}_2^n$.
\end{definition}

We calculate ${n\choose m}_2$ by counting subspaces. We count subspaces by counting bases, then dividing by the number of bases of that subspace. See \cite{van_lint_wilson_2001} for details.

\begin{proposition}\label{q_analog_formula_prop}
         \[{n 
         \choose m}_2 = \prod_{i=0}^{m-1} \frac{2^{n-i}-1}{2^{i+1}
         -1}.\]
\end{proposition}

\begin{proof}(Of Lemma \ref{cval_approximation_helper_lemma})
\begin{enumerate}
    \item We use a formula for  $|\mathcal{L}^n_n|$, (\cite{Gross_2006}, Theorem $20$). It can be derived by the standard technique of counting bases and dividing by overcounts. The formula states

    \[|\mathcal{L}^n_k| = \prod_{i=0}^{k-1} \frac{2^{2n-i}-2^i}{2^{k}-2^i}.\]

    \item We express $Q= \Omega_\frac{n}{2}\mathcal{L}^n_n$, where $\Omega_\frac{n}{2} = |\{ x\in\mathcal{L}^n_n \mid d(y,x)=\frac{n}{2}\}|$ for any fixed $y\in \mathcal{L}^n_n$. Note that $Q$ counts the question pair $(x,y)$ as distinct from $(y,x)$.
    The formula for $\Omega_{\frac{n}{2}}$ is described as ``not hard to see'' in \cite{stanton_1984} (see page 113). One way\footnote{It is also possible to perform these calculations by just counting bases and controlling for overcounting, but this is less direct.} to see this is by modifying Proposition $2.2.2$ of \cite{CHEN19927} to incorporate the isotropic condition.
    
    The same technique can be used to count the degree  $\Delta(G_w)$, as we describe in the next item.
    
    \item
    Let $y\in V(G_w)$. By abuse of notation, $y\in \mathcal{L}^n_n$ where $w\leq y$.

    $\Delta(G_w)$ is the number of $x\in \mathcal{L}^n_n$ such that $w\leq x$ and $\dim (x \cap y) = \frac{n}{2}$.
    
    First, we count the number of possible intersections, $y^\prime \subset y$ with $\dim(y^\prime) = \frac{n}{2}$. Then, we count the number of $x \in \mathcal{L}^n_n$ such that $x\cap y = y^\prime$.
    
Let $y/w$ be the $n-1$-dimensional space that results by the quotient of $y$ by $w$. Choose an $\frac{n}{2}-1$-dimensional subspace of $y/w$ in any of ${n-1 \choose \frac{n}{2} -1}_2$ ways, and lift this $\frac{n}{2}-1$-dimensional subspace to an $\frac{n}{2}$-dimensional subspace  $y^\prime \subset y$ by taking the preimage of the quotient by $w$. 

Let $z$ be a complement of $y^\prime$ in $y^{\prime\perp}$, meaning that $z \cap y^\prime = \{0\}$ and $z+y^\prime = y^{\prime \perp}$. Note that $\dim(y\cap z) = \frac{n}{2}$. Let $z^\prime$ be a complement of $y\cap z$ in $z$.

Any $x$ such that $x \sim y $ can be written uniquely as $x=r + y^\prime$, where $r$ is an isotropic 
$\frac{n}{2}$-dimensional subspace satisfying $r \subset z^\prime$ and $r \cap y = \{0\}$. These subspaces ($r$) are in bijection with linear functions $f:z^\prime \to y\cap z$, under $r \leftrightarrow span(f(z_1)+z_1, f(z_2) +z_2, \dots, f(z_{\frac{n}{2}}) + z_{\frac{n}{2}})$, where $\{z_i\}_{i=1}^{\frac{n}{2}}$ is a basis for $z^\prime$.

The requirement of $r$ being isotropic is equivalent to requiring $\left< f(z_i)+ z_i, f(z_j)+z_j\right> = 0$ for all pairs of basis vectors $z_i,z_j$ of $z^\prime$. We count the number of such functions by counting the number of possibilities for $f(z_i)$, one basis vector at a time. There are $2^{\frac{n}{2}}$ choices for $f(z_1)$. We require $f(z_2)$ to satisfy one linear equation corresponding to the requirement $\left< f(z_1)+ z_1, f(z_2)+z_2\right> = 0$. A single linear constraint reduces the number of choices for $f(z_2)$ to $2^{\frac{n}{2} -1}$. Continuing in this way, we will have $2^{\frac{n}{2}+1-i}$ choices for $f(z_i)$. We calculate the total number of functions by taking the product of choices for each basis vector. This quantity is $2^{\frac{n}{2}+1 \choose 2}$.
\item 
We apply the previous items in this proof and some basic algebra to get the equality. The inequality results from applying the mediant inequality to each factor.
\end{enumerate}
\end{proof}

\end{document}